\documentclass{amsart}
\usepackage{amsmath,amsthm,url,graphicx, latexsym, curves}
\usepackage{amsfonts}
\usepackage{amssymb}

\providecommand{\U}[1]{\protect\rule{.1in}{.1in}}

\newtheorem{theorem}{Theorem}
\newtheorem{lemma}[theorem]{Lemma}
\newtheorem{proposition}[theorem]{Proposition}

\newtheorem{corollary}[theorem]{Corollary}

\theoremstyle{remark}
\newtheorem{remark}[theorem]{Remark}

\numberwithin{equation}{section}

 \renewcommand{\Re}{\operatorname{Re}}
    \renewcommand{\Im}{\operatorname{Im}}

\begin{document}
\title[bipartite entanglement of a random pure states]{Local statistical
properties of Schmidt eigenvalues of bipartite entanglement for a
random pure state}
\author[ Dang-Zheng Liu  and Da-Sheng Zhou ]{ Dang-Zheng Liu and Da-Sheng Zhou}

\address{School of Mathematical Sciences, Peking University, Beijing,
100871, P.R. China}
\email{dzliumath@gmail.com}
\address{Department of Mathematics, University of Macau, Av. Padre Tom\'{a}s
Pereira, Taipa, Macau, P.R. China} \email{zhdasheng@gmail.com}

\begin{abstract}
Consider the model of bipartite entanglement for a random pure
state emerging in quantum information and quantum chaos,
corresponding to the fixed trace Laguerre unitary ensemble (LUE) in
Random Matrix Theory. We focus on correlation functions of Schmidt
eigenvalues for the model and prove universal limits of the
correlation functions in the bulk and also at the soft and hard
edges of the spectrum, as these for the LUE. Further we consider the
bounded trace LUE and obtain the same universal limits.
\end{abstract}

\maketitle

\setcounter{page}{1}

\section{Introduction and main results}

\label{sect1} Quantum entanglement has recently been studied
extensively \cite{Page, Michael,Zyczkowski2001,So,HLW,Nechita,
Majumdar, toda2} due to its central role in quantum information and
quantum computation, which is treated as an indispensable resource
 \cite{Neilsen}. The entanglement of  random
pure quantum states is of much interest in the context of bipartite
entanglement, and statistical properties of such random states are
relevant to quantum chaotic systems, see \cite{Majumdar,HLW} and
references therein.

In the present paper, we consider a bipartite quantum system (a
system with its surrounding environment). Given a composite system
$A\otimes B$ of an $(NM)$-dimensional Hilbert space
$\mathcal{H}^{(NM)}=\mathcal{H}^{(N)}_{A} \otimes
\mathcal{H}^{(M)}_{B}$, let ${|e_{i}^{A}\rangle }_{i=1}^{N}$ and ${%
|e_{j}^{B}\rangle }_{j=1}^{M}$ be two complete orthogonal basis
states for the subsystems $A$ and $B$, respectively.
 Without loss of generality, we assume $N\leq M$. Any quantum state $|\Phi \rangle
\in A\otimes B$ can be expanded as a linear combination
\begin{equation}
|\Phi \rangle =\sum_{i=1}^{N}\sum_{j=1}^{M}X_{ij}|e_{i}^{A}\rangle \otimes
|e_{j}^{B}\rangle   \label{linear combibnation}
\end{equation}%
where these coefficients $X_{ij}\in \mathbb{C}$ form a rectangular
$N\times M$ complex matrix $X=[X_{ij}]$. The composite state $|\Phi
\rangle $ is fully unentangled (separable) if $|\Phi \rangle $ can
be written as a direct
product of two states $|\Phi ^{A}\rangle \in A$ and $|\Phi ^{B}\rangle \in B$%
, i.e., $|\Phi \rangle =|\Phi ^{A}\rangle \otimes \Phi ^{B}\rangle
$, otherwise referred it as an entangled state \cite{Werner}. The
composite state $|\Phi \rangle $ is normalized pure state system if
the density matrix of $|\Phi \rangle $ is given by $\rho =|\Phi
\rangle \langle \Phi |$ satisfying $\text{tr\thinspace }[\rho ]=1$.
The reduced density matrix of the subsystem $A$ by tracing over the
states of the subsystem $B$ is defined \cite{bengtsson,Majumdar} by
\begin{equation}
\rho _{A}\doteq \text{tr\thinspace }_{B}[\rho
]=\sum_{j=1}^{M}\langle e_{j}^{B}|\rho |e_{j}^{B}\rangle
=\sum_{i,j=1}^{N}W_{ij}|e_{i}^{A}\rangle \langle e_{j}^{A}|,
\label{reduces matrix}
\end{equation}%
where $W_{ij}$ are the entries of $N\times N$ square matrix $W=XX^{\dag }$
with $\text{tr\thinspace }W=1$ due to the normalized restriction that $\text{%
tr\thinspace }[\rho ]=1$. Analogously, $\rho _{B}=\text{tr\thinspace }%
_{A}[\rho ]$. It is not difficult to prove that the reduced density matrices
$\rho _{A}$ and $\rho _{B}$ have the same set of non-negative eigenvalues $%
x_{1},x_{2},\ldots ,x_{N}$ (we call them Schmidt eigenvalues of the
quantum state $|\Phi \rangle $) satisfying $\sum_{i=1}^{N}x_{i}=1$.
Let $v_{i}^{A} $ denote the eigenvector of the square matrix $W$
corresponding to the eigenvalue $x_{i}$. Then $\rho _{A}$ can be
expressed as $\rho _{A}=\sum_{i=1}^{N}x_{i}|v_{i}^{A}\rangle \langle
v_{i}^{A}|$. A similar representation holds for $\rho _{B}$. The
composite state $|\Phi \rangle $ has a well-known Schmidt spectral
decomposition \cite{bengtsson}
\begin{equation}
|\Phi \rangle =\sum_{i=1}^{N}\sqrt{x_{i}}|v_{i}^{A}\rangle \otimes
|v_{i}^{B}\rangle .
\end{equation}

A pure state is random if these coefficients $X_{ij}$ are random.
The simplest and most common random state is to choose $X_{ij}$ as
independent and identically distributed Gaussian variables
\cite{Majumdar}. However, the set of complex Wishart matrices
invariant under every unitary transformation but without any other
constraint is referred as complex Wishart ensemble or Laguerre
unitary ensemble \cite{mehta}. Its joint probability density
function (p.d.f.) of $N$ unordered eigenvalue $x_{1},x_{2},\ldots
,x_{N}$ of complex Wishart matrix $W$ is written as
\begin{equation}
P_{N}^{LUE,s}(x_{1},\ldots ,x_{N})=\frac{1}{Z_{s}}\prod_{1\leq
i<j\leq N}|x_{i}-x_{j}|^{2}\prod_{i=1}^{N}x_{i}^{\alpha
}e^{-x_{i}/s}, \label{pdf for weight s}
\end{equation}%
where $s>0$ and $\alpha =M-N\geq 0$. In this present paper, we
extend the scope of the index $\alpha $ and $s$ to $\alpha >-1$ and
$\Re s>0$, respectively. The partition function reads
\begin{equation}
Z_{s}=s^{N(N+\alpha )}\prod_{j=1}^{N}(\Gamma (1+j)\Gamma (\alpha +j)),
\label{partition weight s}
\end{equation}%
calculated in the book of Mehta \cite{mehta}. On the other hand, in
case of a random pure state $|\Phi \rangle $, all the eigenvalues of
$W=XX^{\dag }$ are not quite same as these of complex Wishart matrix
due to the additional constraint that tr$\rho_{A}
\text{=tr\thinspace}W=1$. Thus, the eigenvalues of the reduced
density matrix $\rho _{A}$ are distributed according to (\ref{pdf
for weight s}) in addition to the constraint $\sum_{i=1}^{N}x_{i}=1
$. More precisely,
\begin{equation}
P_{N}^{\delta ,r}(x_{1},\ldots ,x_{N})=\frac{1}{Z_{\delta
}^{r}}\delta (\sum_{i=1}^{N}x_{i}-r)\prod_{1\leq i<j\leq
N}|x_{i}-x_{j}|^{2}\prod_{i=1}^{N}x_{i}^{\alpha },
\label{probability of fixed trace lauguerre}
\end{equation}%
where $\delta (x)$ denotes the Dirac measure, $r=1$ and $\alpha
=M-N$. We refer this ensemble as fixed trace Laguerre unitary
ensemble (FTLUE), following the classic book by Mehta \cite{mehta}
where he referred to fix trace and bounded trace Gaussian ensembles
as restricted trace ensembles (this class of ensembles has been
generalized in \cite{acmv}). We will extend the scope of the index
$r$ to $\Re r>0$. 
It follows from (\ref{partition weight s}%
) that the partition function $Z_{\delta }^{r}$ equals
\begin{equation}
Z_{\delta }^{r}=r^{N(N+\alpha )}\frac{\prod_{j=1}^{N}\Gamma (1+j)\Gamma
(\alpha +j)}{\Gamma (N(N+\alpha ))}.  \label{partition for HS}
\end{equation}%
 When $M=N$, the joint p.d.f. of (\ref{probability of fixed trace lauguerre}) in
\cite{Zyczkowski2001,So} is referred to the ensemble of random
density matrices with respect to the Hilbert-Schmidt metric in the
set $\mathcal{D}_{N} $ of all density matrices of size $N$. It is
worthy of stressing that another ensemble of random density matrices
with respect to the Bures metric is quite distinguished, because its
features support the claim that without any prior knowledge on a
certain density matrix, the optimal way to mimic it is to generate
it at random with respect to the Bures measure \cite{So}.

The study of the eigenvalues of the reduced density matrix
$\rho_{A}=W$ is crucial for understanding and utilizing
entanglement. In principle, all information about the spectral
properties of the subsystem $A$, including its degree of
entanglement, can be  encoded in the p.d.f. of (\ref{probability of
fixed trace lauguerre}). For example, one classic measure of
entanglement is the von Neumann entropy  defined by
$S=-\text{tr\thinspace }\rho_{A} \ln \rho_{A}
=-\sum_{i=1}^{N}x_{i}\ln x_{i}$, which  is a random variable. The average entropy  $%
\langle S\rangle $ is close to $\ln N-\frac{N}{2M}$ for large $N$ when $%
M\geq N$ \cite{Page}. 
Besides, some known results on the FTLUE which are the same in the
limit as these of the LUE include: the  global density
\cite{So,toda2}, the largest eigenvalue distribution \cite{Nechita},
the smallest eigenvalue distribution when $M=N $ \cite{Majumdar}.

In this  paper, we focus on the so-called correlation functions of
the FTLUE. We also consider another closely relevant ensemble:
bounded trace LUE (BTLUE), whose joint p.d.f. for the eigenvalues is
given by
\begin{equation}
P_{N}^{\theta ,r}(x_{1},\ldots ,x_{N})=\frac{1}{Z_{\theta
}^{r}}\theta \big(r-\sum_{i=1}^{N}x_{i}\big)\prod_{1\leq i<j\leq
N}|x_{i}-x_{j}|^{2}\prod_{i=1}^{N}x_{i}^{\alpha },
\label{probability of bounded trace lauguerre}
\end{equation}%
where $\theta (x)$ denotes the Heaviside step function, i.e.,
$\theta (x)=1$ for $x\geq 0$, otherwise $\theta (x)=0$. Note that
the FTLUE or BTLUE bears the same relationship to the LUE that the
micro-canonical ensembles to the canonical ensembles in statistical
mechanics. Section 27 in \cite{mehta}, Mehta posed the ``equivalence
of ensembles" problem whether all local statistical properties of
the eigenvalues between fixed trace and unconstrained random matrix
ensembles  are identical, and further speculated that working out
the eigenvalues spacing distribution for bounded trace ensembles is
much more difficult. Although universal local results have been
obtained for very broad classes of canonical random matrix ensembles
\cite{forrester1,dkmvz,DG,Vanlessen}, only very few results on the
local limit behavior of the correlation functions for the restricted
ensembles (no orthogonal polynomial techniques are available!).
Recently, some progress has been made for fixed trace Gaussian
ensembles \cite{av}, \cite{ggl,gg}, \cite{LZ}. Before we state our
main results, let us first recall that the definitions of
correlation function and some universal results on the LUE.
The $%
n$-point correlation function $R_{n}^{LUE,s}(x_{1},\ldots ,x_{n})$
of the LUE is defined as \cite{mehta}
\begin{equation}
R_{n}^{LUE,s}(x_{1},\ldots ,x_{n})\doteq \frac{N!}{(N-n)!}\int_{\mathbb{R}%
^{N-n}}P_{N}^{LUE,s}(x_{1},\ldots ,x_{N})dx_{n+1}\cdots dx_{N}.
\label{correlation function of LUE}
\end{equation}%
Analogously, the $n$-point correlation function of the FTLUE or
BTLUE is defined as
\begin{equation}
R_{n}^{\phi ,r}(x_{1},\ldots ,x_{n})\doteq \frac{N!}{(N-n)!}\int_{\mathbb{R%
}^{N-n}}P_{N}^{\phi ,r}(x_{1},\ldots ,x_{N})dx_{n+1}\cdots dx_{N},
\label{correlation function of CLUE}
\end{equation}
where $\phi$ denotes $\delta$ or $\theta$. In particular, when
$n=1$,
 $R_{1}^{LUE,s}$ is called  the level density or the density of states. A classical result for the LUE says that
  \begin{equation*}
\lim_{N\rightarrow \infty
}\frac{1}{N}R_{1}^{LUE,\frac{1}{4N}}(x)=\psi (x)\doteq \frac{2}{\pi
}\sqrt{\frac{1-x}{x}}1_{(0,1]}(x),
\end{equation*}%
where the symbol $1_{(0,1]}(x)$ denotes the characteristic function
of the set $(0,1]$, and $\psi (x)$ is the Marchenko-Pastur law
\cite{MP}. 
However, for $n\ge 2$, the study of a finer asymptotics near a point
of the spectrum shows \cite{Nagao,forrester1}: in the bulk, i.e.,
$u\in $ $(0,1),$
\begin{equation*}
\lim_{N\rightarrow \infty }\frac{1}{(N\psi (u))^{n}}R_{n}^{LUE,\frac{1}{4N}%
}(u+\frac{t_{1}}{N\psi (u)},\dots ,u+\frac{t_{n}}{N\psi (u)})=\det [
K_{\sin }(t_{i},t_{j})] _{i,j=1}^{n}
\end{equation*}%
where $K_{\sin }(t_{i},t_{j})=\frac{\sin (\pi (t_{i}-t_{j}))}{\pi
(t_{i}-t_{j})}$ is the so-called sine kernel; at the soft edge ,
\begin{equation*}
\lim_{N\rightarrow \infty }\frac{1}{((2N)^{2/3})^{n}}R_{n}^{LUE,\frac{1}{4N}%
}\left( 1+\frac{t_{1}}{(2N)^{2/3}},\ldots ,1+\frac{t_{n}}{(2N)^{2/3}}\right)
=\det [K_{\mathrm{Airy}}(t_{j},t_{k})]_{j,k=1}^{n}
\end{equation*}%
where
\begin{equation*}
K_{\mathrm{Airy}}(u,v)=\frac{Ai^{\prime }(u)Ai(v)-Ai^{\prime
}(v)Ai(u)}{u-v}
\end{equation*}%
and the Airy function $Ai(x)$ satisfies the equation $Ai^{\prime \prime
}(x)=xAi(x)$
; at the hard edge,
\begin{equation}
\lim_{N\rightarrow \infty }\frac{1}{(16N^{2})^{n}}R_{n}^{LUE,\frac{1}{4N}%
}\left( \frac{t_{1}}{16N^{2}},\ldots ,\frac{t_{n}}{16N^{2}}\right)
=\det (K_{\textrm{J}_{\alpha}}(t_{i},t_{j}))_{i,j=1}^{n},  \notag
\end{equation}%
where
\begin{equation}
K_{\textrm{J}_{\alpha}}(u,v)=\frac{J_{\alpha }(\sqrt{u})\sqrt{v}J_{\alpha }^{\prime }(%
\sqrt{v})-J_{\alpha }(\sqrt{v})\sqrt{u}J_{\alpha }^{\prime }(\sqrt{u})}{%
2(u-v)}
\end{equation}%
and $J_{\alpha }(z)$ denotes the Bessel function of the index $\alpha $.

On the other hand, the limit global density of the FTLUE (see
\cite{So}) and the BTLUE is also the Marchenko-Pastur law, i.e.,
\begin{equation*}
\lim_{N\rightarrow \infty }\frac{1}{N}R_{1}^{\phi ,\frac{N+\alpha }{4}%
}(x)=\psi (x)
\end{equation*}
where $\phi$ denotes $\delta$ or $\theta$. In the case of the BTLUE,
we will prove the claimed result in Sect. 5. Considering
universality in the bulk,  at the soft and hard edges of the
spectrum of the restricted trace LUE,  we have the same local limit
behavior as that for the LUE.

\begin{theorem}
\label{theorem for CLUE: bulk} Let $R_{n}^{\delta ,r}$ be the
$n$-point correlation function of eigenvalues of bipartite
entanglement for a random pure state, defined by (\ref{correlation
function of CLUE}). The  following asymptotic properties hold.

(i) The bulk of the spectrum: for every $u\in (0,1)$ and $t_{i}\in
\mathbb{R},\,1\leq i\leq n$,
\begin{equation}
\lim_{N\rightarrow \infty }\frac{1}{(N\psi (u))^{n}}R_{n}^{\delta ,\frac{%
N+\alpha }{4}}\left( u+\frac{t_{1}}{N\psi (u)},\dots
,u+\frac{t_{n}}{N\psi (u)}\right) =\det [ K_{\sin }(t_{i},t_{j})]
_{i,j=1}^{n},  \notag
\end{equation}%
 uniformly for $t_{1},\ldots, t_{n}$ in compact subsets of $\mathbb{R}$ and for
 $u$ in a compact subset of $(0,1)$.

(ii) The soft edge of the spectrum: for any $f\in
\mathbb{C}_{c}(\mathbb{R}^{n})$, the set of all continuous functions
on $\mathbb{R}^{n} $ with compact support,
\begin{align}
& \lim_{N\rightarrow \infty }\frac{1}{((2N)^{2/3})^{n}}\int_{\mathbb{R}%
^{n}}f(t_{1},\ldots ,t_{n})R_{n}^{\delta ,\frac{N+\alpha }{4}}\left( 1+\frac{%
t_{1}}{(2N)^{2/3}},\ldots ,1+\frac{t_{n}}{(2N)^{2/3}}\right) \,d^{n}t  \notag
\\
& =\int_{\mathbb{R}^{n}}f(t_{1},\ldots ,t_{n})\det [K_{\mathrm{Airy}%
}(t_{j},t_{k})]_{j,k=1}^{n}\,d^{n}t.  \notag
\end{align}

(iii) The hard edge of the spectrum:
\begin{equation}
\lim_{N\rightarrow \infty }\frac{1}{(16N^{2})^{n}}R_{n}^{\delta ,\frac{%
N+\alpha }{4}}\left( \frac{t_{1}}{16N^{2}},\ldots ,\frac{t_{n}}{16N^{2}}%
\right) =\det [K_{\mathrm{J}_{\alpha}}(t_{i},t_{j})]_{i,j=1}^{n},
\notag
\end{equation}
uniformly for $t_{1},\ldots, t_{n}$ in bounded subsets of
$(0,\infty)$.
\end{theorem}

\begin{theorem}
\label{theorem for BTLUE: bulk} Let $R_{n}^{\theta ,r}$ be the
$n$-point correlation function of the bounded trace LUE, defined by
(\ref{correlation function of CLUE}). Let $f\in
\mathbb{C}_{c}(R^{n})$, the set of all continuous functions on
$\mathbb{R}^{n} $ with compact support, the following asymptotic
properties hold.

(i) The bulk of the spectrum: for every $x\in (0,1)$ ,
\begin{align}
& \lim_{N\rightarrow \infty }\frac{1}{(N\psi (x))^{n}}\int_{\mathbb{R}%
^{n}}f(t_{1},\ldots ,t_{n})R_{n}^{\theta ,\frac{%
N+\alpha }{4}}\left( x+\frac{t_{1}}{N\psi (x)},\dots
,x+\frac{t_{n}}{N\psi (x)}\right) \,d^{n}t  \notag
\\
& =\int_{\mathbb{R}^{n}}f(t_{1},\ldots ,t_{n})\det [ K_{\sin
}(t_{i},t_{j})] _{i,j=1}^{n}\,d^{n}t.  \notag
\end{align}


(ii) The soft edge of the spectrum:
\begin{align}
& \lim_{N\rightarrow \infty }\frac{1}{((2N)^{2/3})^{n}}\int_{\mathbb{R}%
^{n}}f(t_{1},\ldots ,t_{n})R_{n}^{\theta,\frac{N+\alpha }{4}}\left( 1+\frac{%
t_{1}}{(2N)^{2/3}},\ldots ,1+\frac{t_{n}}{(2N)^{2/3}}\right)
\,d^{n}t  \notag
\\
& =\int_{\mathbb{R}^{n}}f(t_{1},\ldots ,t_{n})\det [K_{\mathrm{Airy}%
}(t_{j},t_{k})]_{j,k=1}^{n}\,d^{n}t.  \notag
\end{align}

(iii) The hard edge of the spectrum:
\begin{align}
& \lim_{N\rightarrow \infty }\frac{1}{(16N^{2})^{n}}\int_{\mathbb{R}%
^{n}}f(t_{1},\ldots ,t_{n})R_{n}^{\theta,\frac{%
N+\alpha }{4}}\left( \frac{t_{1}}{16N^{2}},\ldots ,\frac{t_{n}}{16N^{2}}%
\right) \,d^{n}t  \notag
\\
& =\int_{\mathbb{R}^{n}}f(t_{1},\ldots ,t_{n})\det
[K_{\mathrm{J}_{\alpha}}(t_{i},t_{j})]_{i,j=1}^{n}\,d^{n}t.  \notag
\end{align}
\end{theorem}
\noindent To the best of our knowledge, this is the first result
about the local properties of correlation functions for the bounded
trace ensembles. Theorems \ref{theorem for CLUE: bulk} and
\ref{theorem for BTLUE: bulk} give an affirmative answer to Mehta's
``equivalence of ensembles" problem in the case of  Laguerre unitary
ensemble. In fact, our method can deal with some more general
ensembles which will be considered in a forthcoming paper.

The plan of the remaining part of our paper is the following.
Sections 2, 3 and 4 are devoted to the proof of Theorem \ref{theorem
for CLUE: bulk}. Section 5 deals with Theorem \ref{theorem for
BTLUE: bulk}. In Sect. 2, the asymptotic behavior of
$R_{n}^{\delta,\frac{N+\alpha}{4}}$ in the bulk of the spectrum is
given  based on the rigorous estimates of the correlation function
$R_{n}^{LUE,\sigma}$ in the complex plane, inspired by \cite{gg}.
Some of the results in \cite{Vanlessen} play an important role on
our proof. In sect. 3, by using the similar method introduced by the
authors in \cite{LZ}, the universality at the soft edge of the
spectrum is proved. In sect. 4, based on a heuristic idea in
\cite{av} where universality at zero is considered for fixed trace
Gauss-type ensembles, 
 the asymptotic behavior at the hard edge is derived. In the last
section, a ``sharp" concentration phenomenon is observed, then local
statistical properties of the eigenvalues between the fixed and
bounded LUEs can be proved to be identical in the limit. So we
extend the results from Theorem \ref{theorem for CLUE: bulk} to
Theorem \ref{theorem for BTLUE: bulk}.

\section{proof of theorem \protect\ref{theorem for CLUE: bulk}: the bulk
of the spectrum}

\subsection{The relation between $R_{n}^{LUE,s}$ and $R_{n}^{\delta,r}$}

For every $\vartheta\in\mathbb{R}$, through this section, we will
denote by $(.)^{\vartheta}$ the function
\begin{equation}
(.)^{\vartheta}:\mathbb{C}\backslash(-\infty,1]\rightarrow\mathbb{C}%
:\,z|\rightarrow\exp\vartheta\log z,
\end{equation}
where log denotes the principle branch of the logarithm. The
constants $C(\mu),C_{1}(\mu)$, $C_{2}(\mu),\,\Omega(\mu)$, depending
on the parameter $\mu$, may change from one line to another line.

Let $s>0$, and let for every $k\geq0$, $\tilde{h}_{k}(x,s)$ be a
polynomial
of degree $k$ with positive leading coefficient such that for $\tilde{h}%
_{k}(x,s)$,
\begin{equation}
\int_{0}^{\infty}\tilde{h}_{k}(x,s)\tilde{h}_{j}(x,s)x^{\alpha}e^{-x/s}dx=%
\delta_{j,k},\,j,\,k=0,1,\ldots.  \label{generate function for h}
\end{equation}
The generalized  Laguerre polynomials
$(\tilde{L}_{j}^{\alpha}(x,s))_{j\geq0}$ with the positive leading
coefficients $\frac{1}{j!}$ are defined by
\begin{equation}
\sum_{j=0}^{\infty}\tilde{L}_{j}^{\alpha}(x,s)w^{j}=(1+sw)^{-\alpha-1}\exp(%
\frac{xw}{1+sw}),  \label{generation function for L}
\end{equation}
and for $i,j=0,1\ldots$,  satisfy the relation
\begin{equation}
\int_{0}^{\infty}e^{-x/s}x^{\alpha}\tilde{L}_{i}^{\alpha}(x,s)\tilde{L}%
_{j}^{\alpha}(x,s)dx=\,s^{\alpha+1+i+j}\frac{\Gamma(i+\alpha+1)}{\Gamma (i+1)%
}\delta_{i,j}.  \label{orthogonal relation for L}
\end{equation}
Hence
\begin{equation}
\tilde{h}_{i}(x,s)=s^{-i-\frac{\alpha+1}{2}}\left( \frac{\Gamma(i+\alpha +1)%
}{\Gamma(i+1)}\right) ^{-1/2}\tilde{L}_{i}^{\alpha}(x,s).
\label{between h and L}
\end{equation}
The functions defined by
\begin{equation}
\tilde{\varphi}_{i}(x,s)=x^{\alpha/2}e^{-x/(2s)}\tilde{h}_{i}(x,s),\,i=0,1,%
\ldots,  \label{varphi}
\end{equation}
form an orthogonal sequence of functions in the Hilbert space $\mathbb{L}%
^{2}(0,\infty)$. Next let us consider the standardized Laguerre
polynomials \cite{szego} with the positive leading coefficient with
respect to the weight $x^{\alpha}e^{-x}$ by the relation
\begin{equation}
L_{i}^{\alpha}(x)=\tilde{L}_{i}^{\alpha}(x,1),\,h_{i}(x)=\tilde{h}%
_{i}(x,1),\,\varphi_{i}(x)=\tilde{\varphi}_{i}(x,1),
\label{relation:LUE with weight s and NO s}
\end{equation}
which satisfy
\begin{equation}
\tilde{L}_{i}^{\alpha}(x,s)=s^{i}L_{i}^{\alpha}(xs^{-1}),\,\tilde{h}%
_{i}(x,s)=\,s^{-(\alpha+1)/2}h_{i}^{\alpha}(x\,s^{-1}),\,\tilde{\varphi}%
_{i}(x,s)=s^{-1/2}\,\varphi_{i}(x\,s^{-1}). \label{Further
relation:LUE with weight s and NO s}
\end{equation}
The following three recurrence formula for
$\tilde{L}_{j}^{\alpha}(x,s)$ holds:
\begin{equation}
j\tilde{L}_{j}^{\alpha}(x,s)=(x\,s^{-1}-2j-\alpha+1)s\tilde{L}_{j-1}^{\alpha
}(x,s)-s^{2}\tilde{L}_{j-2}^{\alpha}(x,s)(j+\alpha-1) \label{third
recuuence realtion}
\end{equation}
where $j=2,3,\ldots$. The $n$-point correlation function
$R_{n}^{LUE,s}$ of the LUE could be expressed as
\begin{equation}
R_{n}^{LUE,s}(x_{1},\dots,x_{n})=\det(\tilde{K}%
_{N}(x_{i},x_{j},s))_{i,j=1}^{n},  \label{Rn of LUE weight s
expressed as KN}
\end{equation}
where
\begin{equation}
\tilde{K}_{N}(x,y,s)=\sum_{k=0}^{N-1}\tilde{\varphi}_{k}(x,s)\tilde{\varphi}%
_{k}(y,s).  \label{KN expressed as polynomial}
\end{equation}
Let
\begin{equation}
K_{N}(x,y)\doteq\tilde{K}_{N}(x,y,1),  \label{KN(x,y,1)}
\end{equation}
from (\ref{Further relation:LUE with weight s and NO s}), we get

\begin{equation}
\tilde{K}_{N}(x,y,s)=s^{-1}K_{N}(x\,s^{-1},y\,s^{-1}).
\label{relation: KN and KsN}
\end{equation}
The Chrisoffel-Darboux formula for kernels $K_{N}$ and
$\tilde{K}_{N}$ reads
\begin{equation}
K_{N}(x,y)=\sqrt{N(N+\alpha)}\frac{\varphi_{N}(x)\varphi_{N-1}(y)-\varphi_{N}(y)\varphi
_{N-1}(x)}{x-y}  \label{CD formula of KN}
\end{equation}
and
\begin{equation}
\tilde{K}_{N}(x,y,s)=\sqrt{N(N+\alpha)}\frac{\tilde{\varphi}_{N}(x,s)\tilde{\varphi}%
_{N-1}(y,s)-\tilde{\varphi}_{N}(y,s)\tilde{\varphi}_{N-1}(x,s)}{x-y}.\label{KNS
expressed as CD}
\end{equation}
Note that the reproducing kernel $K_{N}(x,y)$ has the following
integral representation (Eq.(4.2), \cite{widim}, or Eq.(3.6),
\cite{Johnstone})
\begin{equation}
K_{N}(x,y)=\frac{\sqrt{N(N+\alpha)}}{2}\int_{0}^{+\infty}S_{1}(x+z)S_{2}%
(y+z)+S_{1}(y+z)S_{2}(x+z)dz \label{integral rep of KN}%
\end{equation}
where
\begin{align}
S_{1}(x)  &  =\sqrt{N}\frac{\varphi_{N}(x)}{x}+\sqrt{N+\alpha}\frac
{\varphi_{N-1}(x)}{x},\label{S1}\\
S_{2}(x)  &  =\sqrt{N+\alpha}\frac{\varphi_{N}(x)}{x}+\sqrt{N}\frac
{\varphi_{N-1}(x)}{x}, \label{S2}%
\end{align}
and so, for every $s>0$, the following relation%
\begin{equation}
\label{integral representation of KNs}
K_{N}(\frac{x}{s},\frac{y}{s})=\frac{\sqrt{N(N+\alpha)}}{2}\int_{0}^{+\infty
}S_{1}(\frac{x}{s}+z)S_{2}(\frac{y}{s}+z)+S_{1}(\frac{y}{s}+z)S_{2}(\frac
{x}{s}+z)dz
\end{equation}
holds. The relations (\ref{Further relation:LUE with weight s and NO
s}), (\ref{third recuuence realtion}), (\ref{relation: KN and KsN})
and our extension of the function $(.)^{\vartheta}$ to
$\mathbb{C}\backslash(-\infty,1]$ allow us to continue
$\tilde{L}_{i}(x,s)$, $\tilde{\varphi}_{i}(x,s)$ and
$\tilde{K}_{N}(x,y,s)$ to this domain analytically in the parameter
$s$. The relations (\ref{generation function for L})-(\ref{KNS
expressed as CD}) remain valid under these continuations. So does
(\ref{integral representation of KNs}) whenever this integral is
well-defined.

Next, for $R_{n}^{LUE,s}$ and $R_{n}^{\delta,r},$ we will prove that
one can be expressed by the other. Let us stress that the similar
relation will be frequently used in the proofs of Theorems 1 and 2.
\begin{proposition}
\label{integral equation} Let $R_{n}^{LUE,s}$ and $R_{n}^{\delta
,r}$ be the $n$-point correlation functions, defined by
(\ref{correlation function of LUE}) and (\ref{correlation function
of CLUE}) respectively, then we have the following integral equation
\begin{equation}
R_{n}^{LUE,s}(x_{1},\ldots ,x_{n})=\int_{0}^{\infty }R_{n}^{\delta
,u}(x_{1},\ldots ,x_{n})\,\gamma (\frac{u}{s})s^{-1}du,  \label{LUE
and CLUE}
\end{equation}
where $\gamma (x)$ is defined by
\begin{equation}
\gamma (x)\doteq \frac{1}{\Gamma (N(N+\alpha
))}\,e^{-x}\,x^{N^{2}+N\alpha -1}1_{[0,\infty )}(x).  \label{gamma
def}
\end{equation}
\end{proposition}
\begin{proof}
By (\ref{correlation function of LUE}), for any $f\in
\mathbb{L}^{\infty}(\mathbb{R}^{n})$,
\begin{align}
& \int_{\mathbb{R}^{n}}f(x_{1},\ldots
,x_{n})R_{n}^{LUE,s}(x_{1},\ldots ,x_{n})dx_{1}\ldots
dx_{n}=\frac{1}{Z_{s}}\frac{N!}{(N-n)!}\int_{0}^{\infty
}\int_{\triangle }  \notag  \label{relation} \\
& \times f(uy_{1},\ldots ,uy_{n})\prod_{1\leq i<j\leq
N}|y_{i}-y_{j}|^{2}\prod_{i=1}^{N}y_{i}^{\alpha }e^{-\frac{u}{s}%
}u^{N^{2}+N\alpha -1}dy_{1}\ldots dy_{N}du.  \notag
\end{align}%
Here we make the change of variables: $x_{i}=uy_{i},\,1\leq i\leq N$, and $%
y_{i}$ belongs to the standard N-simplex $\bigtriangleup
=\{(y_{1},\ldots ,y_{N})|\sum_{i=1}^{N}y_{i}=1,\,y_{i}\geq 0\}$. By (\ref%
{correlation function of CLUE}), the right-hand side of the above
equality equals
\begin{align}
& \frac{Z_{\delta}^{1}}{Z_{s}}\int_{0}^{\infty }\int_{\mathbb{R}%
^{n}}f(uy_{1},\ldots ,uy_{n})R_{n}^{\delta ,1}(y_{1},\ldots ,y_{n})e^{-\frac{%
u}{s}}u^{N^{2}+N\alpha -1}dy_{1}\ldots dy_{n}du  \notag \\
& =\frac{Z_{\delta}^{1}}{Z_{s}}\int_{0}^{\infty }\int_{\mathbb{R}%
^{n}}f(y_{1},\ldots ,y_{n})R_{n}^{\delta ,1}(\frac{y_{1}}{u},\ldots ,\frac{%
y_{n}}{u})\frac{1}{u^{n}}e^{-\frac{u}{s}}u^{N^{2}+N\alpha
-1}dy_{1}\ldots
dy_{n}du  \notag \\
& =\frac{Z_{\delta}^{1}}{Z_{s}}\int_{\mathbb{R}^{n}}f(x_{1},\ldots
,x_{n})\int_{0}^{\infty }R_{n}^{\delta ,u}(x_{1},\ldots
,x_{n})e^{-\frac{u}{s}}u^{N^{2}+N\alpha -1}dudx_{1}\ldots dx_{n}.
\notag
\end{align}%
Here we have used the following fact that
\begin{equation}
R_{n}^{\delta ,r}(x_{1},\ldots ,x_{n})=R_{n}^{\delta
,1}(x_{1}r^{-1},\ldots ,x_{n}r^{-1})r^{-n}.  \label{relation for
correlation fun of CLUE}
\end{equation}
Thus we prove that%
\begin{equation}
R_{n}^{LUE,s}(x_{1},\ldots
,x_{n})=\frac{Z_{\delta}^{1}}{Z_{s}}\int_{0}^{\infty }R_{n}^{\delta
,u}(x_{1},\ldots ,x_{n})e^{-\frac{u}{s}}u^{N^{2}+N\alpha -1}du,
\label{relation between LUE and CLUE}
\end{equation}%
It follows from (\ref{partition weight s}) and (\ref{partition for
HS}) that
\begin{equation}
Z_{s}=\Gamma (N(N+\alpha ))s^{N(N+\alpha )}Z_{\delta}^{1}.
\end{equation}
This proves this proposition.
\end{proof}
Note that the relations Eq.(\ref{LUE and CLUE}) and
Eq.(\ref{relation for correlation fun of CLUE}) also hold when
extending the parameter $s$ and $r$ to the domain $\operatorname{Re}
s,\,\operatorname{Re} r>0$. By Proposition 2, set $u=(N+\alpha
+v)/4,\,s=1/4N$, we have
\begin{equation}
R_{n}^{LUE,\frac{1}{4N}}=\int_{0}^{\infty }R_{n}^{\delta ,\frac{N+\alpha +v}{%
4}}N\,\gamma \left( N(N+\alpha +v)\right) dv. \label{LUE and CLUE
after change of variable}
\end{equation}%
We will state the following lemma, which plays a central role in our
proof of universality in the bulk of the spectrum.
\begin{lemma}
\label{lemma: Fourier relation between CLUE and LUE}
\begin{equation}
\int_{0}^{\infty }R_{n}^{\delta ,\frac{N+\alpha +v}{4}}N\,\gamma
\left(
N(N+\alpha +v)\right)\exp(-iyv)dv=\phi _{N}(y)R_{n}^{LUE,\frac{1}{4N(1+iy/N)}%
}\text{,}    \label{fourier transform}
\end{equation}%
where
\begin{equation}
\phi _{N}(y)=e^{iy(N+\alpha )}(1+\frac{iy}{N})^{-(N^{2}+N\alpha
)}\text{.} \label{Phi}
\end{equation}
\end{lemma}
\noindent Note that the function $\phi _{N}(\cdot)$ is the
characteristic function of $N\,\gamma \left( N(N+\alpha +
\cdot)\right)$.
\begin{proof}
Set $(N+\alpha +v)/4=u$, the left-hand side of Eq.(\ref{fourier
transform}) equals
\begin{align}
& \int_{0}^{\infty }R_{n}^{\delta ,u}4N\,\gamma (4Nu)\exp
(-iy(4u-N-\alpha
))du  \notag \\
& =\frac{e^{iy(N+\alpha )}}{\Gamma (N^{2}+N\alpha )}\int_{0}^{\infty
}R_{n}^{\delta ,u}4N\,(4Nu)^{N^{2}+N\alpha -1}\,e^{-4Nu}\exp
(-iy4u)du
\notag \\
& =\int_{0}^{\infty }R_{n}^{\delta ,u}\,e^{-4Nu(1+\frac{iy}{N})}\,(4Nu(1+%
\frac{iy}{N}))^{N^{2}+N\alpha -1}4N(1+\frac{iy}{N})du  \notag \\
& \times \frac{e^{iy(N+\alpha )}}{\Gamma (N^{2}+N\alpha )}(1+\frac{iy}{N}%
)^{-(N^{2}+N\alpha )}  \notag \\
& =e^{iy(N+\alpha )}(1+\frac{iy}{N})^{-(N^{2}+N\alpha )}\,R_{n}^{LUE,\frac{1%
}{4N(1+iy/N)}}=\phi _{N}(y)\,R_{n}^{LUE,\frac{1}{4N(1+iy/N)}}.
\notag
\end{align}
Here we have used Proposition \ref{integral equation}.
\end{proof}
\begin{remark}
By Lemma \ref{lemma: Fourier relation between CLUE and LUE}, using
the inverse Fourier transform, we know that
\begin{equation}
R_{n}^{\delta ,\frac{N+\alpha }{4}}N\,\gamma \left( N(N+\alpha )\right) =%
\frac{1}{2\pi }\int_{-\infty }^{+\infty }\phi _{N}(y)R_{n}^{LUE,\frac{1}{%
4N(1+iy/N)}}dy\text{.}  \label{Inverse FT}
\end{equation}
\end{remark}

\subsection{Estimate of $R_{n}^{LUE,s}$ in the complex plane}

For the convenience of the reader, we will review some basic results
in \cite{Vanlessen}. The function (Eq.(3.38), \cite{Vanlessen})
\begin{equation}
\hat{\psi}(z)\doteq\frac{2}{i\pi
}\frac{(z-1)^{1/2}}{z^{1/2}},\,\text{for}\,z\in \mathbb{C}\backslash
\lbrack 0,1]  \label{complex MP(z)}
\end{equation}%
with principle branches of powers denotes analytic continuation of
the standard Mar\v{c}enko-Pastur law in the domain
$\mathbb{C}\backslash \lbrack 0,1]$. The following two formulas:
\begin{equation}
g(z)=\int_{0}^{1}\log (z-y)\psi (y)dy,\text{for}\,\ \,z\in \mathbb{C}%
\backslash (-\infty ,1]  \label{g_{n}(z)}
\end{equation}%
and
\begin{equation}
\xi (z)=-i\pi \int_{1}^{z}\hat{\psi}(y)dy,\,\text{for}\,z\in
\mathbb{C}\backslash (-\infty ,1]  \label{xi_n(z)}
\end{equation}%
come from Eq.(3.30) and Eq.(3.40) respectively in \cite{Vanlessen}.
The uniqueness of analytic function (Eq.(5.2), \cite{Vanlessen})
shows that
\begin{equation}
2\xi (z)=2g(z)-2z-l,\,\,\text{for}\,z\in \mathbb{C}\backslash
(-\infty ,1]. \label{xi_n(z) and g (z)}
\end{equation}%
Here $l$ is given by Proposition 3.12 in \cite{Vanlessen}.

For $0<\theta<1/2$ and $\beta>0$, the sets $S_{\theta,\beta}$ and
$\overline{S_{\theta,\beta}}$ are defined by
\begin{align}
& S_{\theta,\beta}=\{z\in\mathbb{C}|\,\theta<\Re(z)<1-\theta,\,|\Im z|<\beta\},  \notag \\
&
\overline{S_{\theta,\beta}}=\{z\in\mathbb{C}|\,\theta\leq\Re(z)\leq1-\theta,\,|\Im
z|\leq\beta\}.  \notag
\end{align}

\begin{lemma}
\label{lemma:KN sin kernal in the complex region}Let $\Lambda
(H)=1+iH$. For every $0<\theta<1/2$, there exists a positive number
$H_{0}(\theta)>0$ such that for kernel
$\tilde{K}_{N}(x,y,\frac{1}{4N})$ denoted by (\ref{KN expressed as
polynomial}), and for $A>0$ the relation
\begin{eqnarray}
&&\lim_{N\longrightarrow \infty }\frac{1}{N\psi (x)}\tilde{K}_{N}\left((x+\frac{u%
}{N\psi (x)})\Lambda (H),(x+\frac{v}{N\psi (x)})\Lambda
(H),\frac{1}{4N}\right)
\label{KN limit in the complex domain} \notag\\
&=&\frac{\sin \left(\pi (u-v)\Lambda (H)\frac{\hat{\psi}(x\Lambda (H))}{\psi (x)}\right)}{%
\pi (u-v)\Lambda (H)}\notag
\end{eqnarray}%
holds uniformly for all $x\in \lbrack \theta,1-\theta]$, $0\leq H\leq H_{0}(\theta)$ and $%
|u|,\,|v|\leq A$.
\end{lemma}
\begin{proof}
For any fixed $0<\theta<1/2$, taking $\delta=\theta/4 $ (see Figure
5 in the section  3.8, \cite {Vanlessen}) such that
\begin{equation*}
\lbrack \theta/2,1-\theta/2]\bigcap \partial U_{\delta }=\emptyset
,\,[\theta/2,1-\theta/2]\bigcap \partial \tilde{U}_{\delta
}=\emptyset .
\end{equation*}%
Here $\partial U_{\delta }$ and  $\partial \tilde{U}_{\delta }$
denote the boundary of the two disks $U_{\delta }$ and $
\tilde{U}_{\delta }$ around unity and zero respectively, depicted in
Figure 1 (See Figures 1 and 5, \cite{Vanlessen}). Hence there exists
a positive number $\beta (\frac{\theta}{2})$ such that this set
$\overline{S_{\frac{\theta}{2}},\beta (\frac{\theta}{2})}$ is the
subset of the set $A_{\delta}\bigcup B_{\delta}\bigcup
(\delta,1-\delta)$, described in Figure 1. Let
\begin{equation}
w_{1}=(x+\frac{u}{N\psi (x)})\Lambda (H),\,w_{2}=(x+\frac{v}{N\psi (x)}%
)\Lambda (H).  \label{w(1)w(2)}
\end{equation}
It is not difficult to check that for all $x\in \lbrack
\theta,1-\theta]$ and $|u|,\,|v|\leq A$, there exists $N_{0}$, when
$N\geq N_{0}$ , $\operatorname{\Re} w_{1},\operatorname{\Re}
w_{2}\in
\lbrack \theta/2,1-\theta/2]$. Furthermore, there exists a positive number $H_{0}(\theta)$ such that $0\leq \Im w_{1},\Im w_{2}\leq \beta(\frac{\theta}{2})$ for all $0\leq H\leq H_{0}(\theta)$ if $N$ is sufficiently large. Note that $%
w_{1},w_{2}\in A_{\delta}$. From a series of transformations $%
Y\mapsto U\mapsto T\mapsto S\mapsto R$ (See (3.4), (3.13), (3.36),
(3.51) and (3.99), \cite{Vanlessen}), we obtain
\begin{align}
& \tilde{K}_{N}(w_{1},w_{2},\frac{1}{4N})=4NK_{N}(4Nw_{1},\,4Nw_{2})
\label{relation KN express as Y(x)} \\
& =4N(4Nw_{1})^{\frac{\alpha }{2}}(4Nw_{2})^{\frac{\alpha }{2}}\frac{%
e^{-2Nw_{1}}e^{-2Nw_{2}}}{2i\pi 4N(w_{1}-w_{2})}%
(0,1)Y^{-1}(4Nw_{2})Y(4Nw_{1})%
\begin{pmatrix}
1 \\
0%
\end{pmatrix}
\notag \\
& =\frac{(4N)^{\alpha }w_{1}^{\alpha /2}w_{2}^{\alpha /2}e^{-2N(w_{1}+w_{2})}%
}{2i\pi (w_{1}-w_{2})}(0,1)(4N)^{\frac{\alpha }{2}\sigma _{3}}e^{(\frac{1}{2}%
nl-ng(w_{2}))\sigma _{3}}%
\begin{pmatrix}
1 & 0 \\
-w_{2}^{-\alpha }e^{-2n\xi\left( w_{2}\right) } & 1%
\end{pmatrix}
\notag \\
& \times S(w_{2})^{-1}S\left( w_{1}\right)
\begin{pmatrix}
1 & 0 \\
w_{1}^{-\alpha }e^{-2n\xi\left( w_{1}\right) } & 1%
\end{pmatrix}%
e^{\left( ng\left( w_{1}\right) -\frac{1}{2}nl\right) \sigma
_{3}}\left( 4N\right) ^{-\frac{\alpha }{2}\sigma _{3}}%
\begin{pmatrix}
1 \\
0%
\end{pmatrix}%
,  \notag
\end{align}%
where $\sigma _{3}$ denotes the third Pauli matrix. It follows from
Eq.(3.99), Eq.(3.56) and Theorem 3.32 in \cite{Vanlessen} that the
following relation
\begin{equation}
S(w_{2})^{-1}S\left( w_{1}\right) =I+O(\frac{1}{N})  \label{S(-1)S}
\end{equation}%
holds uniformly for all $w_{1},w_{2}\in \overline{S_{\theta/2,\beta
(\theta/2)}}$. Actually we obtain
\begin{equation*}
S(w_{2})^{-1}S\left( w_{1}\right) =I+O(w_{1}-w_{2}).
\end{equation*}%
Combining (\ref{relation KN express as Y(x)}) and (\ref{S(-1)S}), we
find
\begin{equation}
\tilde{K}_{N}(w_{1},w_{2},\frac{1}{4N})=\frac{w_{1}^{\alpha
/2}w_{2}^{\alpha /2}}{2i\pi (w_{1}-w_{2})}[w_{1}^{-\alpha
}e^{n\xi\left( w_{2}\right) -n\xi\left( w_{1}\right)
}-w_{1}^{-\alpha }e^{n\xi \left( w_{2}\right) -n\xi \left(
w_{1}\right) }].\label{2.38}
\end{equation}%
Here we have used Eq.(\ref{xi_n(z) and g (z)}). By (\ref{w(1)w(2)}),
the following asymptotic behavior
\begin{equation}
n\xi \left( w_{2}\right) -n\xi \left( w_{1}\right) =-N\pi i\int_{(x+%
\frac{u}{N\psi (x)})\Lambda (H)}^{(x+\frac{v}{N\psi (x)})\Lambda (H)}\hat{%
\psi}(y)dy=\pi i\frac{\hat{\psi}(x\Lambda (H))}{\psi (x)}\Lambda (H)(u-v)+O(%
\frac{1}{N})  \notag
\end{equation}%
holds uniformly when $x,u,v,H$ satisfy the assumptions of this
lemma. By (\ref{2.38}), we conclude the proof of this lemma.

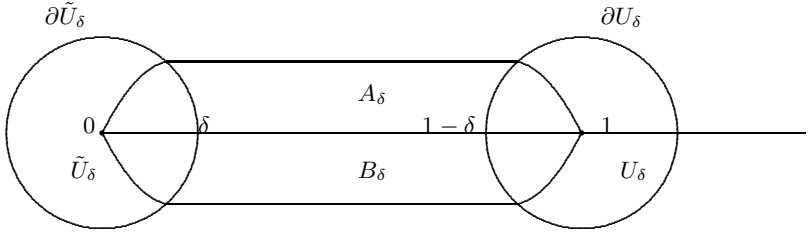
\begin{figure}[t]
\begin{center}
    \setlength{\unitlength}{0.85mm}
    \begin{picture}(130,41.5)(2.5,15)
        \put(20,35){\line(1,0){110}}
           \put(20,35){\bigcircle{30}}
          \put(17,35){\small $0$}
         \put(60,40){\small $A_\delta$}
           \put(60,28){\small $B_\delta$}
           \put(35,35){\small $\delta$}

            \put(20,35){\circle*{.9}}
            \qbezier(20,35)(25,45)(30,46.18)
            \qbezier(20,35)(25,25)(30,23.82)

             \put(15,28){\small $\tilde U_\delta$}

            \put(11,52){\small $\partial \tilde U_\delta$}
        \put(95,35){\bigcircle{30}}
          \put(98,35){\small $1$}
           \put(70,35){\small $1-\delta$}
            \put(101,28){\small $ U_\delta$}

            \put(95,35){\circle*{.9}}
            \qbezier(95,35)(90,45)(85,46.18)
            \qbezier(95,35)(90,25)(85,23.82)

            \put(98,52){\small $\partial U_\delta$}

              \put(30,46.18){\line(1,0){55}}

            \put(30,23.82){\line(1,0){55}}

    \end{picture}
    \caption{ $A_{\delta}$ and $B_{\delta}$ depend on the parameter $\delta$ and the fixed angel $\nu\in(2\pi/3,5\pi/3)$. Details about $\nu$ is described in the section 3.8 and Proposition 3.19 in \cite{Vanlessen}. }
    \label{figure: contour SigmaR}
\end{center}
\end{figure}

\end{proof}

Next, first we will obtain an upper bound about orthogonal
polynomial $h_{N}(z)$ in the complex plane. As a consequence, basing
on the integral representation Eq.(\ref{integral representation of
KNs}), we can derive the upper bound estimate of the reproducing
kernel $\tilde{K}_{N}(\cdot,\cdot,s)$ (See Lemma 8).
\begin{lemma}
\label{lemma: the estimates of hN} For every $\mu>0$, there exist
constants $C(\mu)$ and $\Omega(u)$ such that the two inequalities
\begin{equation}
|h_{N}(4Nz)|\leq
C(\mu)N^{-\frac{\alpha+1}{2}}\Omega(\mu)^{N}|z|^{N+\frac
{1}{2}} \label{hN upper}%
\end{equation}
and
\[
|h_{N-1}(4Nz)|\leq C(\mu)N^{-\frac{\alpha+1}{2}}\Omega(\mu)^{N-1}%
|z|^{N-\frac{1}{2}}%
\]
hold for every $N$ and every $z$ satisfying $\Re z\geq0,\,\Im
z\geq\mu$.
\end{lemma}
\begin{proof}
Theorem 2.4 in \cite{Vanlessen} shows that for $\Re z\geq0$ and $\Im
z\geq\mu $,
\begin{align}
h_{N}(4Nz)  &  =(4Nz)^{-\frac{\alpha}{2}}e^{2Nz}\sqrt{\frac{1}{2N\pi}%
}\nonumber\\
&
\frac{(\Psi(z))^{\frac{\alpha+1}{2}}}{2z^{1/4}(z-1)^{1/4}}\exp(-\pi
iN\int_{1}^{z}\hat{\psi}(s)ds)(1+O(\frac{1}{N})) \label{h(4Nz)}%
\end{align}
where the error term is uniform, and
\begin{equation}
\Psi(z)=2z-1+2z^{1/2}(z-1)^{1/2},\text{ \ for }z\in
\mathbb{C}\backslash\lbrack
0,1]\text{.} \label{Psi(z)}%
\end{equation}
Note that the function $g(z)$ satisfies%
\begin{equation}
g^{\prime}(z)=2(1-\frac{(z-1)^{1/2}}{z^{1/2}}) \label{g'(z)}%
\end{equation}
with the initial condition (See (3.26), (3.28) and (3.35),
\cite{Vanlessen})
\begin{equation}
e^{Ng\left(  z\right)  }=z^{N}+O\left(  z^{N-1}\right).
\label{initial condition}%
\end{equation}
A direct computation shows that
\begin{equation}
g(z)=2z-1-2z^{1/2}(z-1)^{1/2}-2\ln(z^{1/2}-\left(  z-1\right)
^{1/2})-2\ln2\text{.} \label{concrete g_{N}(z)}%
\end{equation}
Note that $l=-2-4\ln2$ (Remark 2.3, \cite{Vanlessen}). It follows
from (\ref{xi_n(z)}) and (\ref{xi_n(z) and g (z)}) that
\begin{equation}
h_{N}(4Nz)=(4Nz)^{-\frac{\alpha}{2}}\sqrt{\frac{1}{2N\pi}}\frac{(\Psi
(z))^{\frac{\alpha+1}{2}}}{2z^{1/4}(z-1)^{1/4}}\frac{e^{2N(z-z^{1/2}\left(
z-1\right)  ^{1/2})}}{(z^{1/2}-(z-1)^{1/2})^{2N}}
\label{h_{N}substitute by g_{N}}
\end{equation}
For a given analytic branch,
$(z^{1/2})^{2}=z,\,(z-1)=((z-1)^{1/2})^{2}$. Set $w=z^{1/2}$, which
is a univalent analytic function in
$\mathbb{C}\backslash(-\infty,1]$. Let $x(w)$ be the root of this
equation: $x+\frac{1}{x}=2w$ satisfying $|x|<1$. Hence
\begin{equation}
x(w)=w-\sqrt{w^{2}-1},
\end{equation}
where the function $\sqrt{w^{2}-1}$ is a univalent analytic function
in
$\mathbb{C}\backslash(-\infty,1]$ satisfying the relation $(\sqrt{w^{2}%
-1})^{2}>0$ for $w>1$. For every $0<r<1$, the equation $|x(w)|=r$
defines the ellipse
\begin{equation}
\frac{(\operatorname{Re}w)^{2}}{(\frac{1}{r}+r)^{2}}+\frac{(\operatorname{Im}%
w)^{2}}{(\frac{1}{r}-r)^{2}}=\frac{1}{4}, \label{ellipse}%
\end{equation}
and $|x(w)|\leq r$ denotes the outside of this ellipse $E_{r}.$ Thus
for $w\in
E_{r}$,%
\begin{equation}
\label{2.47} |x(w)|=|w-\sqrt{w^{2}-1}|\leq r<1,
\end{equation}
and so
\begin{equation}
\operatorname{Re}(w-\sqrt{w^{2}-1})^{2}\leq r^{2}\text{.} \label{rexx<rr}%
\end{equation}
By Eq.(\ref{ellipse}), for $w\in E_{r}$ we get
\begin{equation}
\min_{|x(w)|=r}|w|=\frac{1}{2}(\frac{1}{r}-r)\text{.}
\label{the minimuu of |w|}%
\end{equation}
Hence,
\begin{align}
\frac{1}{|w-\sqrt{w^{2}-1}|}  &  =\frac{1}{|x(w)|}=|2w-x(w)|\\
&  <1+2|w|=|w|(2+\frac{1}{|w|})\leq|w|(2+\frac{2}{(\frac{1}{r}-r)})\text{.}%
\end{align}
Further, we have
\begin{align}
|w^{2}-1|  &  =|w-1||w+1|=|\frac{1}{2}(x+\frac{1}{x})-1||\frac{1}{2}%
(x+\frac{1}{x})+1|\nonumber\\
&
=|\frac{1}{2x}|^{2}|x-1|^{2}|x+1|^{2}\geq\frac{1}{4r^{2}}(1-r)^{2}.
\label{lower bound of |w2-1|}%
\end{align}
Apply (\ref{2.47})--(\ref{lower bound of |w2-1|}) to
(\ref{h_{N}substitute by g_{N}}), we find that for every fixed
$0<r<1$, the function $|h_{N}(4Nz)|$ could be controlled by
\begin{align}
|h_{N}(4Nz)|  &
=|(4Nz)^{-\frac{\alpha}{2}}|\sqrt{\frac{1}{2N\pi}}\frac
{1}{2|z^{1/4}(z-1)^{1/4|}}\frac{|e^{N(z^{1/2}-\left(  z-1\right)  ^{1/2})^{2}%
}|e^{N}}{|(z^{1/2}-(z-1)^{1/2})|^{2N+\alpha+1}}\nonumber\\
&  =|(4Nw^{2})^{-\frac{\alpha}{2}}|\sqrt{\frac{1}{2N\pi}}\frac{1}%
{2|w|^{1/2}|w^{2}-1|^{1/4}}\frac{|e^{N(x(w))^{2}}|e^{N}}{|x(w)|^{2N+\alpha+1}%
}\nonumber\\
&  \leq
C_{1}(r)N^{-\frac{\alpha+1}{2}}L(r)^{N}|w|^{2N+1}=C_{1}(r)N^{-\frac
{\alpha+1}{2}}L(r)^{N}|z|^{N+\frac{1}{2}}\text{.} \label{hN upper bound}%
\end{align}
Let $z_{N}=\frac{N-1}{N}z$. Note that by the convexity of
$C\backslash E_{r}$, $z_{N}\in E_{r}$ if $z\in E_{r}$. Thus we have
\begin{align}
|h_{N-1}(4Nz)|  &  =|h_{N-1}(4(N-1)z_{N})|\leq
C_{1}(r)(N-1)^{-\frac{\alpha
+1}{2}}L(r)^{N-1}|z_{N}|^{N-\frac{1}{2}}\nonumber\\
&  \leq C_{2}(r)N^{-\frac{\alpha+1}{2}}L(r)^{N}|z|^{N-\frac{1}{2}}.
\label{hN-1 upper bound}%
\end{align}
Here we have used the following fact that%
\[
\lim_{N\rightarrow\infty}(\frac{N-1}{N})^{N+\frac{1}{2}}=e^{-1}.
\]
According to the assumption about $z$ of this lemma, it follows from
$w=|z|^{1/2}e^{i\frac{\arg z}{2}}$ that $\Re w\geq0,\,\Im w\geq0$.
Note that $w^{2}=(\Re w)^{2}-(\Im w)^{2}+2i\Re w\Im w$. The
assumption $\Re z\geq0$ implies that $\Re w\geq\Im w$. For every
$\mu>0$, let
\[
r=\sqrt{\frac{\mu+2}{2}}-\sqrt{\frac{\mu}{2}}\text{,}%
\]
then $0<r<1$ is a solution of the equation
$(\frac{1}{r}-r)^{2}=2\mu$. Hence $\operatorname{Im}z\geq\mu$
implies that
\[
2(\Re w)^{2}\geq2\Re w\Im w=\Im w^{2}=\Im z\geq\frac{1}{2}(\frac{1}{r}%
-r)^{2}\text{,}%
\]
so $\Re w\geq\frac{1}{2}(\frac{1}{r}-r)$, which establishes this
fact that $w\in E_{r}$. It follows from Eq.(\ref{hN upper bound})
and Eq.(\ref{hN-1 upper bound}) that for every $\mu>0$, there exist
some constants
$C(\mu)$ and $\Omega(\mu)$ such that%
\begin{align}
|h_{N}(4Nz)|  &  \leq C(\mu)N^{-\frac{\alpha+1}{2}}\Omega(\mu)^{N}%
|z|^{N+\frac{1}{2}}\text{,}\label{hN about mu}\\
|h_{N-1}(4Nz)|  &  \leq C(\mu)N^{-\frac{\alpha+1}{2}}\Omega(\mu)^{N-1}%
|z|^{N-\frac{1}{2}}\text{.} \label{hN-1 about mu}%
\end{align}\end{proof}

\begin{lemma}
\label{Lemma: KN(z1,z2)} For every $\mu>0$, there exist $C(\mu)$ and
$\Omega(\mu)$ such that
\begin{equation}
\frac{1}{N}|\tilde{K}_{N}(z_{1},z_{2},\frac{1}{4N})|\leq
C(\mu)\Omega (\mu)^{2N-2}\Gamma(2P+1)|\Im z_{1}|^{P+\frac{1}{2}}|\Im
z_{1}|^{P+\frac{1}{2}}
\label{upper bound of KN(z1,z2,1/4N)}%
\end{equation}
holds for all complex numbers $z_{i}$, $i=1,\,2$ satisfying $\Re
z_{i}>0$,
$\Im z_{i}\geq\mu$, and $N\geq\frac{1}{4\mu}$. Here $P=N+[\frac{\alpha-1}%
{2}]+1$, $[x]$ denotes the integer part of a real number $x$.
\end{lemma}
\begin{proof}
By (\ref{integral representation of KNs}), one finds that
\begin{align}
& \frac{1}{N}\tilde{K}_{N}(z_{1},z_{2},\frac{1}{4N})=4K_{N}(4Nz_{1}%
,4Nz_{2})=4\frac{\sqrt{N(N+\alpha)}}{2}4N\nonumber\\
&  \times\int_{0}^{+\infty}S_{1}(4N(z_{1}+\theta))S_{2}(4N(z_{2}%
+\theta))+S_{1}(4N(z_{2}+\theta))S_{2}(4N(z_{1}+\theta))d\theta\text{.}
\label{the aim of KN(4N)}
\end{align}
It follows from Eq.(\ref{varphi}),  Eq.(\ref{S1}) and Lemma
\ref{lemma: the estimates of hN} that
\begin{align*}
&
|S_{1}(4N(z_{1}+\theta))|\leq\sqrt{N}\frac{|\varphi_{N}(4N(z_{1}+\theta
))|}{|4N(z_{1}+\theta)|}+\sqrt{N+\alpha}\frac{|\varphi_{N-1}(4N(z_{1}%
+\theta))|}{|4N(z_{1}+\theta)|}\\
&  \leq\sqrt{N}\frac{|4N(z_{1}+\theta)|^{\alpha/2}e^{-2N\operatorname{Re}%
(z_{1}+\theta)}C(\mu)N^{-\frac{\alpha+1}{2}}\Omega(\mu)^{N}|z_{1}%
+\theta|^{N+\frac{1}{2}}}{|4N(z_{1}+\theta)|}\\
&  +\sqrt{N+\alpha}\frac{|4N(z_{1}+\theta)|^{\alpha/2}e^{-2N\operatorname{Re}%
(z_{1}+\theta)}C(\mu)N^{-\frac{\alpha+1}{2}}\Omega(\mu)^{N-1}|z_{1}%
+\theta|^{N-\frac{1}{2}}}{|4N(z_{1}+\theta)|}\text{.}%
\end{align*}
Here we have used this fact that $\theta\geq0$, $\ \operatorname{Re}%
(z_{1}+\theta)>0$ and $\operatorname{Im}(z_{1}+\theta)\geq\mu$. A
direct calculation tells us that
\begin{align*}
&  |S_{1}(4N(z_{1}+\theta))|\leq4^{\alpha/2-1}N^{-1}C(\mu)\Omega(\mu
)^{N}|z_{1}+\theta|^{N+\frac{\alpha-1}{2}}e^{-2N\operatorname{Re}(z_{1}%
+\theta)}\\
&  +4^{\alpha/2-1}N^{-3/2}\sqrt{N+\alpha}C(\mu)\Omega(\mu)^{N-1}|z_{1}%
+\theta|^{N+\frac{\alpha-1}{2}-1}e^{-2N\operatorname{Re}(z_{1}+\theta)}\\
&  =4^{\alpha/2-1}N^{-1}C(\mu)\Omega(\mu)^{N-1}e^{-2N\operatorname{Re}%
(z_{1}+\theta)}|z_{1}+\theta|^{N+\frac{\alpha-1}{2}}(\Omega(\mu)+\frac
{1}{|z_{1}+\theta|}\sqrt{\frac{N+\alpha}{N}})\\
&  \leq N^{-1}C(\mu)\Omega(\mu)^{N-1}e^{-2N\operatorname{Re}(z_{1}%
+\theta)}|z_{1}+\theta|^{N+\frac{\alpha-1}{2}}\text{.}%
\end{align*}
Note that $|z_{1}+\theta|^{-1}\leq1/\mu$. Analogously, we have
\begin{equation}
|S_{2}(4N(z_{2}+\theta))|\leq N^{-1}C(\mu)\Omega(\mu)^{N-1}%
e^{-2N\operatorname{Re}(z_{2}+\theta)}|z_{2}+\theta|^{N+\frac{\alpha-1}{2}%
}\text{.}%
\end{equation}
Applying the two estimates of $S_{1}$ and $S_{2}$ to (\ref{the aim of KN(4N)}%
), we find that%
\begin{align}
&  \frac{1}{N}|\tilde{K}_{N}(z_{1},z_{2},\frac{1}{4N})|\leq%
C(\mu)\Omega(\mu)^{2N-2}\nonumber\\
&  \int_{0}^{+\infty}e^{-2N\operatorname{Re}(z_{1}+\theta)}%
e^{-2N\operatorname{Re}(z_{2}+\theta)}|z_{1}+\theta|^{N+\frac{\alpha-1}{2}%
}|z_{2}+\theta|^{N+\frac{\alpha-1}{2}}d\theta\nonumber\\
&  \leq C(\mu)\Omega(\mu)^{2N-2}\left( \int
_{0}^{+\infty}e^{-4N\operatorname{Re}(z_{1}+\theta)}|z_{1}+\theta
|^{2(N+\frac{\alpha-1}{2})}d\theta\right)  ^{\frac{1}{2}}%
\label{KN controlled by integral}\\
&  \times\left(
\int_{0}^{+\infty}e^{-4N\operatorname{Re}(z_{2}+\theta
)}|z_{2}+\theta|^{2(N+\frac{\alpha-1}{2})}d\theta\right)
^{\frac{1}{2}}.\nonumber
\end{align}
Let
$\{\frac{\alpha-1}{2}\}=\frac{\alpha-1}{2}-[\frac{\alpha-1}{2}]$.
For every complex number $z$ satisfying $\operatorname{Re}z\geq0$,
$\operatorname{Im}z\geq\mu$, we obtain
\begin{align}
\int_{0}^{+\infty} e^{-4N\operatorname{Re}(z+\theta)}  &  |z+\theta
|^{2(N+\frac{\alpha-1}{2})}d\theta\leq\mu^{2(\{\frac{\alpha-1}{2}\}-1)}\int
_{0}^{+\infty}e^{-4N\operatorname{Re}(z+\theta)}|z+\theta|^{2P}d\theta
\nonumber\\
&
=\mu^{2(\{\frac{\alpha-1}{2}\}-1)}\int_{\operatorname{Re}z}^{+\infty
}e^{-4N\eta}|\eta+i\operatorname{Im}z|^{2P}d\eta\nonumber\\
&
=\mu^{2(\{\frac{\alpha-1}{2}\}-1)}\int_{\operatorname{Re}z}^{+\infty
}e^{-4N\eta}|\eta^{2} +(\operatorname{Im}z)^{2}|^{P}d\eta\nonumber\\
&
\leq\mu^{2(\{\frac{\alpha-1}{2}\}-1)}\int_{0}^{+\infty}e^{-4N\eta}|\eta
^{2}+(\operatorname{Im}z)^{2}|^{P}d\eta\nonumber\\
&
=\mu^{2(\{\frac{\alpha-1}{2}\}-1)}(\operatorname{Im}z)^{2P+1}\int_{0}
^{+\infty}e^{-4N\eta\operatorname{Im}z}|\eta^{2}+1|^{P}d\eta\nonumber\\
&
\leq\mu^{2(\{\frac{\alpha-1}{2}\}-1)}(\operatorname{Im}z)^{2P+1}\int
_{0}^{+\infty}e^{-4N\eta\mu}|\eta^{2}+1|^{P}d\eta\nonumber\\
&
\leq\mu^{2(\{\frac{\alpha-1}{2}\}-1)}(\operatorname{Im}z)^{2P+1}\int
_{0}^{+\infty}e^{-4N\eta\mu}|\eta+1|^{2P}d\eta\nonumber\\
&
=\mu^{2(\{\frac{\alpha-1}{2}\}-1)}(\operatorname{Im}z)^{2P+1}e^{4N\mu}
(\frac{1}{4N\mu})^{2P+1}\int_{4N\mu}^{+\infty}e^{-y}y^{2P}dy\text{.} \label{7}%
\end{align}
Here we have used the inequality: $\eta^{2}+1\leq(\eta+1)^{2}$ for
$\eta\geq 0$. We also notice that (P.17, \cite{feller})
\[
\frac{1}{\Gamma(m+1)}\int_{c}^{+\infty}x^{m}e^{-x}dx=e^{-c}(1+\frac{c}
{1!}+\frac{c^{2}}{2!}+\cdots+\frac{c^{m}}{m!})\text{, for }m\in
N\text{,
}c>0\text{.}%
\]
Then we have
\begin{align}
&  e^{4N\mu}(\frac{1}{4N\mu})^{2P+1}\int_{4N\mu}^{+\infty}e^{-y}%
y^{2P}dy=e^{4N\mu}(\frac{1}{4N\mu})^{2P+1}\Gamma(2P+1)\int_{4N\mu}^{+\infty
}\frac{e^{-y}y^{2P}}{\Gamma(2P+1)}dy\nonumber\\
&
=e^{4N\mu}(\frac{1}{4N\mu})^{2P+1}\Gamma(2P+1)\left(e^{-4N\mu}(1+\frac{4N\mu
}{1!}+\frac{(4N\mu)^{2}}{2!}+\cdots+\frac{(4N\mu)^{2P}}{(2P)!})\right)\nonumber\\
&  \leq\frac{1}{4N\mu}\Gamma(2P+1)e\text{.} \label{uB of exp-y}%
\end{align}
Here we have used the inequality: $4N\mu\geq1$. It follows from
(\ref{KN controlled by integral}), (\ref{7}) and (\ref{uB of exp-y})
that
\begin{align}
&\frac{1}{N}|\tilde{K}_{N}(z_{1},z_{2},\frac{1}{4N})|\leq\notag\\
&\frac{1}{\sqrt{N}
}C(\mu)\Omega(\mu)^{2N-2}\mu^{\{\frac{\alpha-1}{2}\}-1}\frac{e}
{4N\mu}\Gamma(2P+1)(\operatorname{Im}z_{1})^{P+\frac{1}{2}}(\operatorname{Im}
z_{2})^{P+\frac{1}{2}} \text{.}\nonumber
\end{align}
Hence we complete the proof of this lemma.
\end{proof}

\begin{corollary}
\label{corollary} Let $R_{n}^{LUE,s}$ be n-point correlation
function of the LUE, defined by (\ref{correlation function of LUE}),
for all real $u,\,t_{i}$, $i=1,\ldots,n$, if $N$ is sufficiently
large, the following inequality
\begin{align*}
&  \left\vert \frac{1}{N^{n}}R_{n}^{LUE,\frac{1}{4N}}\left(  (u+\frac{t_{1}%
}{N\psi(u)})\Lambda(H),\dots,(u+\frac{t_{n}}{N\psi(u)})\Lambda(H)\right)
\right\vert \\
&  \leq n!C(\theta)^{n}\Omega(\theta)^{(2N-2)n}\Gamma(2P+1)^{n}H^{(2P+1)n}%
(1-\frac{\theta}{2})^{(2P+1)n}%
\end{align*}
holds in all $u\in\lbrack\theta,1-\theta]$ , $|t_{i}|\leq A$ and
$H\geq H_{0}(\theta)$. Here $0<\theta<1$, $A>0$ and $H_{0}(\theta)$
as in Lemma \ref{lemma:KN sin kernal in the complex region}.
\end{corollary}
\begin{proof}
For any $u\in\lbrack\theta,1-\theta]$, there exists $N^{\prime}$
such that $|t_{i}/N\psi(u)|\leq\theta/2$ for all $N>N^{\prime}$.
Hence we have
\begin{align*}
\Re((u+\frac{t_{i}}{N\psi(u)})\Lambda(H))  &  >\theta/2>0
\end{align*}
and
\begin{align*}
\Im((u+\frac{t_{i}}{N\psi(u)})\Lambda(H))  &  \geq H\frac{\theta
}{2}\geq H_{0}(\theta)\frac{\theta}{2}\text{.}%
\end{align*}
It follows from Lemma \ref{Lemma: KN(z1,z2)} that
\begin{align*}
&
\frac{1}{N}|\tilde{K}_{N}((u+\frac{t_{i}}{N\psi(u)})\Lambda(H),(u+\frac
{t_{j}}{N\psi(u)})\Lambda(H),\frac{1}{4N})|\\
&  \leq
C(\theta)\Omega(\theta)^{2N-2}\Gamma(2P+1)|H(u+\frac{t_{i}}{N\psi
(u)})|^{P+\frac{1}{2}}|H(u+\frac{t_{j}}{N\psi(u)})|^{P+\frac{1}{2}}\\
&  \leq C(\theta)\Omega(\theta)^{2N-2}\Gamma(2P+1)H^{2P+1}(1-\frac{\theta}%
{2})^{2P+1}.
\end{align*}
By the definition of determinant and Eq.(\ref{Rn of LUE weight s
expressed as KN}), we have
\begin{align*}
&  \left\vert \frac{1}{N^{n}}R_{n}^{LUE,\frac{1}{4N}}\left(  (u+\frac{t_{1}%
}{N\psi(u)})\Lambda(H),\dots,(u+\frac{t_{n}}{N\psi(u)})\Lambda(H)\right)
\right\vert \\
&  =\frac{1}{N^{n}}|\det(\tilde{K}_{N}((u+\frac{t_{i}}{N\psi(u)}%
)\Lambda(H),u+\frac{t_{j}}{N\psi(u)},\frac{1}{4N}))_{i,j=1}^{n}|\\
&  \leq n!\max_{|t_{i}|\text{,}|t_{j}|\leq A}|\frac{1}{N}\tilde{K}%
_{N}((u+\frac{t_{i}}{N\psi(u)})\Lambda(H),u+\frac{t_{j}}{N\psi(u)},\frac
{1}{4N})|^{n}\\
&  \leq n!C(\theta)^{n}\Omega(\theta)^{(2N-2)n}\Gamma(2P+1)^{n}H^{(2P+1)n}%
(1-\frac{\theta}{2})^{(2P+1)n}\text{.}%
\end{align*}
We complete the proof of this corollary.
\end{proof}

\subsection{Proof of Theorem 1 in the bulk of the spectrum}
Now, we turn to the proof of Theorem \ref{theorem for CLUE: bulk} in
the bulk.
\begin{proof}
From the Stirling's formula $\Gamma (z)=\sqrt{2\pi}\exp (-z)z^{z-1/2}(1+O(\frac{1}{|z|}))$ if $%
|z|\rightarrow \infty ,\,|\arg (z)|<\pi $, we know that
\begin{equation}
\lim_{N\rightarrow \infty }N\,\gamma \left( N(N+\alpha +v)\right) =\frac{1}{%
\sqrt{2\pi }}e^{-\frac{v^{2}}{2}}.  \label{N gamma}
\end{equation}%
Particularly taking $v=0$, then $N\,\gamma \left(
N(N+\alpha)\right)\rightarrow 1/\sqrt{2\pi}$ as
$N\rightarrow\infty$. A straightforward calculation proves that
\begin{equation}
\lim_{N\rightarrow \infty }\int_{-\infty }^{\infty }\phi
_{N}(y)dy=\sqrt{2\pi }.  \label{the integral of varphi(y)}
\end{equation}
By Eq.(\ref{Inverse FT}), it suffices to show that
\begin{align*}
\lim_{N\rightarrow\infty} &
\frac{1}{(N\psi(u))^{n}}\int_{-\infty}^{+\infty
}\phi_{N}(y)R_{n}^{LUE,\frac{1}{4N(1+iy/N)}}(u+\frac{t_{1}}{N\psi(u)}%
,\dots,u+\frac{t_{n}}{N\psi(u)})dy\\
&  =\sqrt{2\pi}\det\left(  \frac{\sin(\pi(t_{i}-t_{j}))}{\pi(t_{i}-t_{j}%
)}\right)  _{i,j=1}^{n}%
\end{align*}
holds uniformly in $u$, $t_{1},\ldots,t_{n}$ satisfying the
assumptions of this theorem. Note that
\[
\phi_{N}(-y)=\overline{\phi_{N}(y)}\text{, }R_{n}^{LUE,\frac{1}{4N(1-iy/N)}%
}=\overline{R_{n}^{LUE,\frac{1}{4N(1+iy/N)}}}\text{, }%
\]
and
\[
\lim_{N\rightarrow\infty}\int_{0}^{\infty}\operatorname{Re}\phi_{N}%
(y)dy=\frac{\sqrt{2\pi}}{2}\text{.}%
\]
It is enough to show that%
\begin{align*}
&
\lim_{N\rightarrow\infty}\frac{1}{(N\psi(u))^{n}}\int_{0}^{+\infty
}\operatorname{Re}[\phi_{N}(y)R_{n}^{LUE,\frac{1}{4N(1+iy/N)}}(u+\frac{t_{1}%
}{N\psi(u)},\dots,u+\frac{t_{n}}{N\psi(u)})]dy\\
&  =\lim_{N\rightarrow\infty}\int_{0}^{\infty}\operatorname{Re}\phi
_{N}(y)dy\det\left(
\frac{\sin(\pi(t_{i}-t_{j}))}{\pi(t_{i}-t_{j})}\right)
_{i,j=1}^{n}\text{.}%
\end{align*}
Making the change of variables: $y/N=H$, we find
\begin{align*}
&  \frac{1}{(N\psi(u))^{n}}\int_{0}^{+\infty}\operatorname{Re}[\phi
_{N}(y)R_{n}^{LUE,\frac{1}{4N(1+iy/N)}}(u+\frac{t_{1}}{N\psi(u)},\dots
,u+\frac{t_{n}}{N\psi(u)})]dy\\
&  =\frac{N}{(N\psi(u))^{n}}\int_{0}^{+\infty}dH\\
&  \operatorname{Re}\left[
\phi_{N}(NH)R_{n}^{LUE,\frac{1}{4N}}\left(
(u+\frac{t_{1}}{N\psi(u)})\Lambda(H),\dots,(u+\frac{t_{n}}{N\psi(u)}%
)\Lambda(H)\right)  \Lambda(H)^{n}\right]  \\
&
=\frac{N}{(N\psi(u))^{n}}(\int_{0}^{\bar{H}}+\int_{\bar{H}}^{+\infty
})=I_{1}+I_{2}%
\end{align*}
where $\bar{H}\in(0,H_{0}(\theta))$ fixed and $H_{0}(\theta)$ as in
Lemma 5. Here we have used the following fact that
\begin{equation}
R_{n}^{LUE,\sigma s}(x_{1},\ldots ,x_{n})=R_{n}^{LUE,s}(x_{1}\sigma
^{-1},\ldots ,x_{n}\sigma ^{-1})\sigma ^{-n} \label{relation for
correlation fun of LUE}
\end{equation}%
holds for any $\Re\sigma >0$ and $s>0$. On the other hand, let%
\begin{align*}
&  \int_{0}^{\infty}\operatorname{Re}\phi_{N}(y)dy\det\left(
\frac{\sin
(\pi(t_{i}-t_{j}))}{\pi(t_{i}-t_{j})}\right)  _{i,j=1}^{n}\\
&
=(\int_{0}^{\bar{H}}+\int_{\bar{H}}^{+\infty})\operatorname{Re}\phi
_{N}(NH)dH\det\left(
\frac{\sin(\pi(t_{i}-t_{j}))}{\pi(t_{i}-t_{j})}\right)
_{i,j=1}^{n}=J_{1}+J_{2}.
\end{align*}
Next we will prove that when $N\rightarrow\infty$,%
\begin{equation}
I_{1}-J_{1}\rightarrow0\text{, }I_{2}\rightarrow0\text{,
}J_{2}\rightarrow
0\text{.}%
\end{equation}
By Lemma 5, for every $0<\theta<1$, $A>0$ the following relation
\begin{align*}
&  \lim_{N\rightarrow\infty}\frac{1}{(N\psi(u))^{n}}R_{n}^{LUE,\frac{1}{4N}%
}\left(
(u+\frac{t_{1}}{N\psi(u)})\Lambda(H),\dots,(u+\frac{t_{n}}{N\psi
(u)})\Lambda(H)\right)  \\
&  =\det\left( \frac{\sin\pi(t_{i}-t_{j})\Lambda(H)\frac{\hat{\psi}(u\Lambda(H))}%
{\psi(u)}}{\pi(t_{i}-t_{j})\Lambda(H)}\right)
_{i,j=1}^{n}\triangleq\tilde {K}(u,H)
\end{align*}
holds uniformly in all $u\in\lbrack\theta,1-\theta]$, $|t_{i}|\leq
A$ and $H\in\lbrack0,\bar{H}]$. The uniform continuity of the
function $\tilde {K}(u,H)\Lambda(H)^{n}$ with respect to $H$ means
that for every $\epsilon$, there exists an $H(\epsilon)\leq
H_{0}(\theta)$ such that
\begin{equation}
|\tilde{K}(u,H)\Lambda(H)^{n}-\tilde{K}(u,0)|<\epsilon
\end{equation}
holds for all $0\leq H\leq H(\epsilon)$. We choose
$\bar{H}=H(\epsilon)$. Note that
\begin{equation}
|\phi_{N}(y)|=\exp(-\frac{N(N+\alpha)}{2}\ln(1+\frac{y^{2}}{N^{2}}))\text{,}%
\end{equation}
and so
\begin{equation}
N\int_{0}^{\bar{H}}|\operatorname{Re}\phi_{N}(NH)|dH\leq2\int_{0}^{+\infty
}\exp(-\frac{y^{2}}{2})dy\text{.}%
\end{equation}
Hence, for large $N$ the difference $I_{1}-J_{1}$ can be controlled
by
\begin{align*}
|I_{1}-J_{1}| &  \leq N\int_{0}^{\bar{H}}\left\vert
\operatorname{Re}\left[ \phi_{N}(NH)\left(
\frac{1}{(N\psi(u))^{n}}R_{n}^{LUE,\frac{1}{4N}}-\tilde
{K}(u,H)\right)  \Lambda(H)^{n}\right]  \right\vert dH\\
&  +N\int_{0}^{\bar{H}}\left\vert \operatorname{Re}\left[  \phi_{N}%
(NH)\tilde{K}(u,H)\Lambda(H)^{n}-\tilde{K}(u,0)\right]  \right\vert dH\\
&  \leq\epsilon4\int_{0}^{+\infty}\exp(-\frac{y^{2}}{2})dy.
\end{align*}
The uniform boundedness of the function $\tilde{K}(u,0)$ with
respect to $u\in\lbrack\theta,1-\theta]$ implies that as
$N\rightarrow\infty$, the following relation
\[
|J_{2}|\leq|\tilde{K}(u,0)|N\int_{\bar{H}}^{+\infty}|\operatorname{Re}\phi
_{N}(NH)|dH\leq|\tilde{K}(u,0)|2\int_{N\bar{H}}^{+\infty}\exp(-\frac{y^{2}}%
{2})dy\rightarrow0
\]
holds uniformly in all $u$. It follows from Corollary \ref{corollary} that%
\begin{align*}
|I_{2}| &  \leq\frac{N}{\psi(u)^{n}}n!C(\theta)^{n}\Omega(\theta
)^{(2N-2)n}\Gamma(2P+1)^{n}(1-\frac{\theta}{2})^{(2P+1)n}\\
&\times\int_{\bar{H}%
}^{+\infty}|\operatorname{Re}\left[
\phi_{N}(NH)\Lambda(H)^{n}\right]
|H^{(2P+1)n}dH\\
&  \leq\frac{N}{\psi(\theta)^{n}}n!C(\theta)^{n}\Omega(\theta)^{(2N-2)n}%
(1-\frac{\theta}{2})^{(2P+1)n}\\
&  \times\Gamma(2P+1)^{n}\int_{\bar{H}}^{+\infty}e^{-\frac{N(N+\alpha)}{2}%
\ln(1+H^{2})}(1+H^{2})^{\frac{n}{2}}H^{(2P+1)n}dH\\
&  \leq e^{C^{\prime\prime}N}\Gamma(2P+1)^{n}\int_{\bar{H}}^{+\infty}%
e^{-\frac{N(N+\alpha)}{2}\ln(1+H^{2})}(1+H^{2})^{n(P+1)}dH\\
&  \leq e^{C^{\prime\prime}N}\Gamma(2P+1)^{n}\int_{\bar{H}}^{+\infty
}e^{-(\frac{N(N+\alpha)}{2}-n(P+1)-1)\ln(1+H^{2})}\frac{1}{1+H^{2}}dH\\
&  \leq e^{C^{\prime\prime}N}\Gamma(2P+1)^{n}e^{-(\frac{N(N+\alpha)}%
{2}-n(P+1)-1)\ln(1+\bar{H}^{2})}\int_{0}^{+\infty}\frac{1}{1+H^{2}}dH\text{.}%
\end{align*}
Here the constant $C^{\prime\prime}$ depends on $\theta$ and $n$. On
the other hand, we notice that
\[
\Gamma(2P+1)^{n}=\sqrt{2\pi}\exp(n(2P+1/2)\ln(2P+1)-2P-1)\sim
O(e^{n2N\ln
(2N)})\text{,}%
\]
thus $I_{2}\rightarrow0$ as $N\rightarrow\infty.$ We complete the
proof of this theorem in the bulk.
\end{proof}

\section{proof of theorem \ref{theorem for CLUE: bulk}: the soft edge of the spectrum}
From Eq.(\ref{LUE and CLUE}), we get
\begin{equation}
\label{fixed trace soft edge}
R_{n}^{LUE,\frac{1}{4N}}(x_{1},x_{2},\cdots ,x_{n})=\int_{0}^{+\infty
}R_{n}^{\delta ,\frac{N+\alpha }{4}}(\frac{x_{1}}{u},\frac{x_{2}}{u},\cdots ,%
\frac{x_{n}}{u})\frac{1}{u^{n}}\gamma (N_{\alpha}u)N_{\alpha}du,
\end{equation}
where $N_{\alpha}=N(N+\alpha)$. Next, We  prove a more refined
asymptotic result than Eq.(\ref{N gamma}).
\begin{lemma}
\label{concentrationlemma1}Let $\{b_{N}\}$ be a sequence such that $%
b_{N}\rightarrow 0$ but $Nb_{N}/\sqrt{\ln N}\rightarrow \infty $ as $%
N\rightarrow \infty $, then we have
\begin{equation}
\label{soft edge: lemma 1}
\int_{0}^{\infty }\gamma (N_{\alpha}u)N_{\alpha}\,du=\int_{u_{-}}^{u_{+}}\gamma (4Nu)4N\,du+O(e^{-\frac{1}{2}%
(Nb_{N})^{2}(1+o(1))}),
\end{equation}%
where $u_{\pm }=\frac{N+\alpha }{4}(1\pm b_{N})$.
\end{lemma}
\begin{proof}
Divide the left hand side of the Eq.(\ref{soft edge: lemma 1}) into
three parts
\begin{equation}
\left(\int_{0}^{1-b_{N}}+\int_{1-b_{N}}^{1+b_{N}}+\int_{1+b_{N}}^{\infty
}\right)\gamma (N(N+\alpha )u)N(N+\alpha )\,du.
\end{equation}
First consider $\int_{0}^{u_{-}}\gamma (4Nu)4N\,du$. By $\gamma ^{\prime
}(4Nu)=0$, we get the maximum point
\begin{equation}
u_{max}=\frac{(N+\alpha )}{4}>u_{-}
\end{equation}
for sufficiently large $N$. Note that the function $\gamma (4Nu)$ is
monotonically increasing when $u\in (0,u_{-}]$, thus
\begin{equation}
\int_{0}^{u_{-}}\gamma (4Nu)4N\,du\leq \gamma
(N_{\alpha}(1-b_{N}))N_{\alpha }(1-b_{N}).
\end{equation}%
It follows from (\ref{gamma def}) that
\begin{equation}
\gamma (N_{\alpha }(1-b_{N}))N_{\alpha
}(1-b_{N})=\frac{(1-b_{N})^{N_{\alpha }}e^{-N_{\alpha
}(1-b_{N})}}{N_{\alpha }{}^{-N_{\alpha }}\Gamma (N_{\alpha })}.
\end{equation}
Using Stirling's formula, we get
\begin{align}
\label{integralestimation1}
&\ln [\gamma (N_{\alpha }(1-b_{N}))N_{\alpha }(1-b_{N})]\\
& =N_{\alpha }\ln (1-b_{N})-N_{\alpha }(1-b_{N})+N_{\alpha }\ln
N_{\alpha
}-[(N_{\alpha }-\frac{1}{2})\ln N_{\alpha }-N_{\alpha }+O(1)]  \notag \\
& =N_{\alpha }\ln (1-b_{N})+N_{\alpha }b_{N}+\frac{1}{2}\ln N_{\alpha }+O(1)
\notag \\
& =N_{\alpha }[-b_{N}-\frac{1}{2}b_{N}^{2}+O(b_{N}^{3})]+N_{\alpha }b_{N}+%
\frac{1}{2}\ln N_{\alpha }+O(1)  \notag \\
& =-\frac{N_{\alpha }}{2}b_{N}^{2}+N(N+\alpha )\,O(b_{N}^{3})+\frac{1}{2}\ln
N_{\alpha }+O(1)  \notag \\
& =-\frac{N_{\alpha }}{2}b_{N}^{2}[1+O(b_{N})+\frac{1}{N_{\alpha }b_{N}^{2}}%
\ln N_{\alpha }+O(\frac{1}{N_{\alpha }b_{N}^{2}})]  \notag \\
& =-\frac{1}{2}N^{2}b_{N}^{2}[1+2\frac{\ln N}{N^{2}b_{N}^{2}}+o(1)+O(\frac{1%
}{\ln N})]=-\frac{1}{2}N^{2}b_{N}^{2}(1+o(1)).\notag
\end{align}
In the last two equalities we used $b_{N}\rightarrow 0$ and $Nb_{N}/\sqrt{%
\ln N}\rightarrow \infty $ as $N\rightarrow \infty $. Next, we
estimate $\int_{u_{+}}^{\infty }\gamma (4Nu)4N\,du$. Notice that the
unique maximum point of $u^{2}\,\gamma (4Nu)\,$ satisfies
\begin{equation}
\tilde{u}_{max}=\frac{N_{\alpha }+1}{4N}<u_{+}\notag
\end{equation}
for sufficiently large $N$, and $u^{2}\,\gamma (4Nu)$ is
monotonically decreasing when $u\in [u_{+},\infty ),$ thus
\begin{align}
&\int_{u_{+}}^{\infty }\gamma (4Nu)4N\,du =\int_{u_{+}}^{\infty}u^{-2}(u^{2}\,\gamma (4Nu)4N\,)\,du \notag\\
&\leq u_{+}^{2}\gamma (4Nu_{+})4N\int_{u_{+}}^{\infty
}u^{-2}\,du=4Nu_{+}\gamma (4Nu_{+})=\gamma (N_{\alpha
}(1+b_{N}))N_{\alpha}(1+b_{N}).\notag
\end{align}
Similarly, we can get
\begin{equation}
\ln [\gamma (N_{\alpha }(1+b_{N}))N_{\alpha }(1+b_{N})]=-\frac{1}{2}
N^{2}b_{N}^{2}(1+o(1)).\label{integralestimation2}
\end{equation}
Combing (\ref{integralestimation1}) and (\ref{integralestimation2}), this completes the proof.
\end{proof}

\begin{remark}
In Lemma \ref{concentrationlemma1}, let us take $b_{N}=N^{-\kappa},
\kappa\in (0,1)$. It is a well-known fact that the scaling at the
soft edge of the spectrum is proportional to $N^{-2/3}$, thus we can
choose $\kappa>2/3$ and give a very close approximation of
correlation functions near the radial sharp cutoff point. Then using
known results about the unconstrained ensembles, we obtain
Airy-kernel for the fixed trace ensembles. Such arguments can also
deal with Bessel-kernel at the hard edge. However, it seems to be
insufficient for proving universality in the  bulk. The main
difficulty is that the ``rate" index $\kappa$ has been rather sharp,
in the sense that it cannot be replaced with a larger number than 1.
\end{remark}

\begin{proof}{(Theorem \ref{theorem for CLUE: bulk}: the soft edge)}\\
For $f\in C_{c}(\mathbb{R}^{n})$, From (\ref{fixed trace soft edge})
one finds that
\begin{align}
& \frac{1}{((2N)^{2/3})^{n}}\int_{\mathbb{R}^{n}}f(t_{1},\cdots
,t_{n})R_{n}^{LUE,\frac{1}{4N}}\left(1+\frac{t_{1}}{(2N)^{2/3}},\cdots ,1+\frac{%
t_{n}}{(2N)^{2/3}}\right)\,d^{n}t \label{three terms} \\
& =\frac{1}{((2N)^{2/3})^{n}}\int_{\mathbb{R}^{n}}f(t_{1},\cdots
,t_{n})\int_{0}^{+\infty }  \notag \\
& R_{n}^{\delta ,\frac{N+\alpha }{4}}\left(u^{-1}(1+\frac{t_{1}}{(2N)^{2/3}}%
),\cdots ,u^{-1}(1+\frac{t_{1}}{(2N)^{2/3}})\right)\frac{1}{u^{n}}\gamma
(N_{\alpha }u)N_{\alpha }du\,d^{n}t \notag\\
& =\frac{1}{((2N)^{2/3})^{n}}\int_{\mathbb{R}^{n}}\left(\int_{0}^{1-b_{N}}+\int_{1-b_{N}}^{1+b_{N}}+\int_{1-b_{N}}^{\infty}\right)\notag\\
&f \left((2N)^{2/3}(u-1)+uy_{1},\cdots , (2N)^{2/3}(u-1)+uy_{n}\right)\notag\\
&\times R_{n}^{\delta,\frac{N+\alpha}{4}}\left(1+\frac{y_{1}}{(2N)^{2/3}},\cdots ,1+\frac{y_{n}}{(2N)^{2/3}}\right)\gamma (N_{\alpha}u)N_{\alpha
}\,du\,d^{n}y  \notag \\
& \doteq I_{1}+I_{2}+I_{3}.\notag
\end{align}
Here we have used the change of variables:
\begin{equation}\label{changeofvariables}
t_{i}=(2N)^{2/3}(u-1)+uy_{i},i=1,\ldots ,n.
\end{equation}
Choose $b_{N}=N^{-\gamma }$ for fixed $\gamma \in (\frac{2}{3},1)$. By Lemma
\ref{concentrationlemma1}, we get
\begin{align}
& |I_{1}|\leq \frac{\Vert f\Vert _{\infty }}{((2N)^{2/3})^{n}}\int_{\mathbb{R%
}^{n}}\int_{0}^{1-b_{N}}\gamma (N_{\alpha }u)N_{\alpha }R_{n }^{\delta,\frac{N+\alpha}{4}
}(1+\frac{y_{1}}{(2N)^{2/3}},\cdots ,1+\frac{y_{n}}{(2N)^{2/3}})\,du\,d^{n}y
\label{threeterm-1}\notag \\
& \leq \frac{\Vert f\Vert _{\infty }}{((2N)^{2/3})^{n}}\int_{\mathbb{R}%
^{n}}R_{n }^{\delta,\frac{N+\alpha}{4} }(1+\frac{t_{1}}{(2N)^{2/3}},\cdots ,1+\frac{t_{n}}{%
(2N)^{2/3}})\,d^{n}y\int_{0}^{1-b_{N}}\gamma (N_{\alpha }u)N_{\alpha }du
\notag \\
& =\Vert f\Vert _{\infty }\frac{N!}{(N-n)!}%
O(e^{-1/2(Nb_{N})^{2}(1+o(1))})=o(1).  \notag
\end{align}
Similarly, again by Lemma \ref{concentrationlemma1}, we have $|I_{3}|=o(1).$
Thus \begin{equation}(\ref{three terms})=I_{2}+o(1) \label{three = I2+o(1)}.\end{equation}
Let
\begin{align}
& I_{2}^{\prime }\doteq \frac{1}{((2N)^{2/3})^{n}}\int_{\mathbb{R}%
^{n}}\int_{1-b_{N}}^{1+b_{N}}f(y_{1},\cdots ,y_{n})  \notag \\
& \,\times R_{n}^{\delta ,\frac{N+\alpha }{4}}\left( 1+\frac{y_{1}}{%
(2N)^{2/3}},\cdots ,1+\frac{y_{n}}{(2N)^{2/3}}\right) \,\gamma (N_{\alpha
}u)N_{\alpha }\,du\,d^{n}y,  \notag \\
& =\frac{1}{((2N)^{2/3})^{n}}\int_{\mathbb{R}^{n}}f(y_{1},\cdots
,y_{n})R_{n}^{\delta ,\frac{N+\alpha }{4}}\left( 1+\frac{y_{1}}{(2N)^{2/3}}%
,\cdots ,1+\frac{y_{n}}{(2N)^{2/3}}\right) \,d^{n}y\,  \notag \\
& \times \int_{1-b_{N}}^{1+b_{N}}\gamma (N_{\alpha }u)N_{\alpha }\,du.
\notag
\end{align}%
Next we will prove that
\begin{equation}
\lim_{N\rightarrow \infty }|I_{2}-I_{2}^{\prime }|=0.\label{I2-I2'}
\end{equation}
Since $f\in C_{c}(\mathbb{R}^{n})$ and
\begin{equation}
(2N)^{2/3}(u-1)+uy_{i}\overset{N\rightarrow \infty }{\longrightarrow }%
y_{i},i=1,\ldots ,n,
\end{equation}%
we can choose a ball $B_{R}$ of the radius $R$ in $\mathbb{R}^{n}$ centered
at zero such that $supp(f)\subset B_{R}$ and
\begin{align*}
\{& ((2N)^{2/3}(u-1)+uy_{1},\ldots ,(2N)^{2/3}(u-1)+uy_{n})| \\
& (y_{1},\ldots ,y_{n})\in supp(f),1-b_{N}\leq u\leq 1+b_{N}\}\subset B_{R}.
\end{align*}%
From $f\in C_{c}(\mathbb{R}^{n})$, given $\epsilon >0$, there exists some $%
\delta (\epsilon )>0$ such that $|f(x_{1},\ldots
,x_{n})-f(y_{1},\ldots ,y_{n})|<\epsilon ,$ whenever $\Vert
(x_{1},\ldots ,x_{n})-(y_{1},\ldots ,y_{n})\Vert <\delta (\epsilon
)$, $\forall \,(x_{1},\ldots
,x_{n}),(y_{1},\ldots ,y_{n})\in B_{R}$. On the other hand, there exist $%
N_{0}$ independent of $(y_{1},\ldots ,y_{n})\in B_{R}$ such that
\begin{align*}
& \Vert ((2N)^{2/3}(u-1)+uy_{1},\ldots
,(2N)^{2/3}(u-1)+uy_{n})-(y_{1},\cdots ,y_{n})\Vert  \\
& \leq \sqrt{n}|u-1|(R+(2N)^{2/3})\leq \sqrt{n}N^{-\gamma
}(R+(2N)^{2/3})\leq \delta (\epsilon )
\end{align*}%
for $N>N_{0}$. Therefore, $\forall \,(y_{1},\ldots ,y_{n})\in supp(f)$
\begin{equation}
|f((2N)^{2/3}(u-1)+uy_{1},\ldots ,(2N)^{2/3}(u-1)+uy_{n})-f(y_{n},\ldots
,y_{n})|<\epsilon .
\end{equation}%
Furthermore, we get
\begin{align*}
& |I_{2}-I_{2}^{\prime }|\leq \frac{\epsilon }{((2N)^{2/3})^{n}} \\
& \int_{B_{R}}\int_{1-b_{N}}^{1+b_{N}}\frac{1%
}{((2N)^{2/3})^{n}} R_{n}^{\delta ,\frac{N+\alpha }{4}}(1+%
\frac{y_{1}}{(2N)^{2/3}},\cdots ,1+\frac{y_{n}}{(2N)^{2/3}})\gamma
(N_{\alpha }u)N_{\alpha }du\,d^{n}y \\
& =\epsilon \int_{1-b_{N}}^{1+b_{N}}\gamma (N_{\alpha }u)N_{\alpha }du\frac{1%
}{((2N)^{2/3})^{n}}\int_{B_{R}}R_{n}^{\delta ,\frac{N+\alpha }{4}}(1+\frac{%
y_{1}}{(2N)^{2/3}},\cdots ,1+\frac{y_{n}}{(2N)^{2/3}})d^{n}y \\
& \leq \epsilon \,(1+o(1))C_{R}.
\end{align*}%
Here $C_{R}$ is a constant and we have used Lemma \ref{boundlemma} below. Hence the relation
(\ref{I2-I2'}) holds. Combining (\ref{three terms}) and (\ref{three = I2+o(1)}), we have $(\ref{three terms})=I^{\prime}_{2}+o(1)$ for sufficiently large $N$, more precisely,
\begin{align}
& \frac{1}{((2N)^{2/3})^{n}}\int_{\mathbb{R}^{n}}f(t_{1},\cdots
,t_{n})R_{n}^{LUE,\frac{1}{4N}}\left(1+\frac{t_{1}}{(2N)^{2/3}},\cdots ,1+\frac{%
t_{n}}{(2N)^{2/3}}\right)\,d^{n}t\notag\\
& =\frac{1}{((2N)^{2/3})^{n}}\int_{\mathbb{R}^{n}}f(y_{1},\cdots
,y_{n})R_{n}^{\delta ,\frac{N+\alpha }{4}}\left( 1+\frac{y_{1}}{(2N)^{2/3}}
,\cdots ,1+\frac{y_{n}}{(2N)^{2/3}}\right) \,d^{n}y(1+o(1))+o(1),\notag
\end{align}
this proves the anticipated result.
\end{proof}
\begin{lemma}
\label{boundlemma} For any fixed $R>0$, let $B_{R}$ be the ball of the
radius $R$ in $\mathbb{R}^{n}$ centered at zero. There exists some constant $%
C_{R}$ such that
\begin{equation}
\frac{1}{((2N)^{2/3})^{n}}\int_{B_{R}}R_{n}^{\delta ,\frac{N+\alpha }{4}}(1+%
\frac{y_{1}}{(2N)^{2/3}},\cdots ,1+\frac{y_{n}}{(2N)^{2/3}})d^{n}y\leq C_{R}.
\end{equation}
\end{lemma}
\begin{proof}
Given $0<\delta <R$, there exists $N_{0}(R,\delta )$ such that when $%
N>N_{0}(R,\delta )$ and $(y_{1},\cdots ,y_{n})\in B_{R}$,
\begin{equation}
\sum_{i=1}^{n}((2N)^{2/3}(u-1)+uy_{i})^{2}<(R+\delta )^{2}
\end{equation}%
where $u\in \lbrack 1-N^{-\gamma },1+N^{-\gamma }]$, $2/3<\gamma <1$. Let $%
\eta \in (0,1)$ be a real number and $\phi (t)$ be a smooth decreasing
function on $[0,\infty )$ such that $\phi (t)=1$ for $t\in \lbrack
0,R+\delta )$ and $\phi (t)=0$ for $t\geq (1+\eta )(R+\delta )$. Set $%
\varphi (x_{1},\cdots ,x_{n})=\phi (\Vert (x_{1},\cdots ,x_{n})\Vert )$ for $%
(x_{1},\cdots ,x_{n})\in \mathbb{R}^{n}$. For $N>N_{0}$, we have
\begin{align}
& \frac{1}{((2N)^{2/3})^{n}}\int_{B_{R}}R_{n}^{\delta ,\frac{N+\alpha }{4}%
}(1+\frac{y_{1}}{(2N)^{2/3}},\cdots ,1+\frac{y_{n}}{(2N)^{2/3}})\,d^{n}y
\label{boundlemmainequality} \\
& \leq \frac{1}{((2N)^{2/3})^{n}}\int_{\mathbb{R}^{n}}\varphi
((2N)^{2/3}(u-1)+uy_{1},\cdots ,(2N)^{2/3}(u-1)+uy_{n})  \notag \\
& R_{n}^{\delta ,\frac{N+\alpha }{4}}(1+\frac{y_{1}}{(2N)^{2/3}},\cdots ,1+%
\frac{y_{n}}{(2N)^{2/3}})\,d^{n}y  \notag
\end{align}%
Multiplying by $\gamma (N_{\alpha }u)N_{\alpha }$ and then integrating (\ref%
{boundlemmainequality}) with respect to $u$ on $[1-b_{N},1+b_{N}]$, one
obtains
\begin{align*}
& \int_{1-b_{N}}^{1+b_{N}}\gamma (N_{\alpha }u)N_{\alpha }du\,\frac{1}{%
((2N)^{2/3})^{n}}\int_{B_{R}}R_{n\beta }^{\delta }(1+\frac{y_{1}}{(2N)^{2/3}}%
,\cdots ,1+\frac{y_{n}}{(2N)^{2/3}})\,d^{n}y \\
& \leq \frac{1}{((2N)^{2/3})^{n}}\int_{\mathbb{R}^{n}}%
\int_{1-b_{N}}^{1+b_{N}}\varphi \Big((2N)^{2/3}(u-1)+uy_{1},\cdots
,(2N)^{2/3}(u-1)+uy_{n}\Big) \\
& \gamma (N_{\alpha }u)N_{\alpha }R_{n}^{\delta ,\frac{N+\alpha }{4}}(1+%
\frac{y_{1}}{(2N)^{2/3}},\cdots ,1+\frac{y_{n}}{(2N)^{2/3}})du\,d^{n}y \\
& =\frac{1}{((2N)^{2/3})^{n}}\int_{\mathbb{R}^{n}}\varphi (t_{1},\cdots
,t_{n})R_{n}^{LUE,\frac{1}{4N}}(1+\frac{t_{1}}{(2N)^{2/3}},\cdots ,1+\frac{%
t_{n}}{(2N)^{2/3}})\,d^{n}t+o(1).
\end{align*}%
Here we make use of Eq.(\ref{three = I2+o(1)}). This completes the
proof of this lemma.
\end{proof}

\section{proof of theorem \protect\ref{theorem for CLUE: bulk}: the hard
edge of the spectrum}

First we prove that the $n$-point correlation function $R_{n}^{\delta,r}$
could be expressed as the inverse Laplace transform of $R_{n}^{LUE,\frac{1}{%
4N}}$, which is slightly different from Eq.(\ref{fourier
transform}).

\begin{proposition}
\label{disintegrationforRTE}Let $R_{n}^{LUE,\frac{1}{4N}}$ and $%
R_{n}^{\delta ,r}$ be the $n$-point correlation functions of
eigenvalues for the LUE and  FTLUE, respectively. Then we have the
relation
\begin{equation}
R_{n}^{\delta ,r}(x_{1},\ldots ,x_{n})=\frac{\Gamma (N_{\alpha })}{%
r^{N_{\alpha }-1}(4N)^{n}}\mathcal{L}^{-1}[t^{-(N_{\alpha }-n)}R_{n}^{LUE,%
\frac{1}{4N}}(\frac{t}{4N}x_{1},\ldots ,\frac{t}{4N}x_{n})](r)  \notag
\label{basicrelation3}
\end{equation}%
where $\mathcal{L}^{-1}[h(t)](x)$ is the inverse Laplace transform of a
function $h(t)$, and $N_{\alpha}=N(N+\alpha)$.
\end{proposition}

\begin{proof}
For $h\in L^{\infty }(\mathbb{R}^{N})$, let $<h(\cdot)>$ and
$<h(\cdot)>_{\delta}$ denote that the ensemble average is taken in
the LUE and the FTLUE, respectively. Consider the integral
\begin{equation}
I[h]=\int_{\mathbb{R}^{N}}h(x_{1},\ldots ,x_{N})\,\delta \big(%
r-\sum_{i=1}^{N}x_{i}\big)\prod_{i=1}^{N}x_{i}^{\alpha
}\prod_{j<k}|x_{j}-x_{k}|^{2}d^{N}\,x.  \notag
\end{equation}
Making the change of variables: $x_{j}=4Ny_{j},j=1,\ldots ,N$, we have
\begin{equation}
I[h]=(4N)^{N_{\alpha }}\int_{\mathbb{R}^{N}}h(4Nx_{1},\ldots
,4Nx_{N})\,\delta \big(r-4N\sum_{i=1}^{N}x_{i}\big)\prod_{i=1}^{N}x_{i}^{%
\alpha }\prod_{j<k}|x_{j}-x_{k}|^{2}d^{N}\,x.  \notag
\end{equation}
Multiply both sides by $e^{-t\,r}$ and integrate on $r$ from $0$ to $\infty $%
, we get
\begin{align}
& \int_{0}^{\infty }e^{-t\,r}I[h]d\,r=(4N)^{N_{\alpha }}\int_{\mathbb{R}%
^{N}}h(4Nx_{1},\ldots ,4Nx_{N})  \notag \\
&\times\exp \big(-t4N\sum_{i=1}^{N}x_{i}\big)\,\prod_{i=1}^{N}x_{i}^{\alpha
}\prod_{j<k}|x_{j}-x_{k}|^{2}d^{N}x  \notag \\
& =(\frac{4N}{t})^{N_{\alpha }}\int_{\mathbb{R}^{N}}h(\frac{4N}{t}%
x_{1},\ldots ,\frac{4N}{t}x_{N})\exp \big(-4N\sum_{i=1}^{N}x_{i}\big)%
\,\prod_{i=1}^{N}x_{i}^{\alpha }\prod_{j<k}|x_{j}-x_{k}|^{2}d^{N}x  \notag \\
& =Z_{1/4N}\,(\frac{4N}{t})^{N_{\alpha }}<h(\frac{4N}{t}\cdot )>.  \notag
\end{align}%
Here we have made the change of variables $x_{j}=t^{-1}y_{j},j=1,\ldots ,N$,
where $t^{-1}$ denotes the principal branch of the power for complex
variable $t$. Using the inverse Laplace transform, we have
\begin{equation}
I[h]=Z_{1/4N }\mathcal{L}^{-1}[(\frac{4N}{t})^{N_{\alpha }}<h(\frac{4N}{t}
\cdot )>](r).  \notag
\end{equation}
Notice that
\begin{equation*}
\mathcal{L}^{-1}[t^{-\gamma }](x)=\frac{x^{\gamma -1}}{\Gamma (\gamma )}%
\theta (x),\,\mathrm{Re}(\gamma )>0.  \label{propertyofILT}
\end{equation*}%
The ensemble average $<h>_{\delta }$ reads
\begin{equation*}
<h>_{\delta }=\frac{I[h]}{I[1]}=\frac{\Gamma (N_{\alpha })}{r^{N_{\alpha }-1}%
}\mathcal{L}^{-1}[t^{-N_{\alpha }}<h(\frac{4N}{t}\cdot )>](r).
\end{equation*}%
In particular, taking
\begin{equation*}
h(x_{1},\ldots ,x_{N})=\sum_{1\leq i_{1}<\cdots <i_{n}\leq
N}f(x_{i_{1}},\ldots ,x_{i_{n}}),
\end{equation*}%
then we find that
\begin{align}
& \int_{\mathbb{R}^{n}}f(x_{1},\ldots ,x_{n})R_{n}^{\delta ,r}(x_{1},\ldots
,x_{n})\,d^{n}x  \notag \\
& =\frac{\Gamma (N_{\alpha })}{r^{N_{\alpha }-1}}\frac{1}{2\pi i}
\int_{-i\infty +0^{+}}^{-i\infty +0^{+}}dt\,e^{rt}t^{-N_{\alpha }}  \notag \\
&\times\int_{\mathbb{R}^{n}}f(\frac{4N}{t}x_{1},\ldots ,\frac{4N}{t}%
x_{n})R_{n}^{LUE, \frac{1}{4N}}(x_{1},\ldots ,x_{n})\,d^{n}x  \notag \\
& =\frac{\Gamma (N_{\alpha })}{r^{N_{\alpha }-1}(4N)^{n}}\frac{1}{2\pi i}
\int_{-i\infty +0^{+}}^{-i\infty +0^{+}}dt\,e^{rt}t^{-(N_{\alpha }-n)}
\notag \\
&\times \int_{\mathbb{R}^{n}}f(x_{1},\ldots ,x_{n})R_{n}^{LUE,\frac{1}{4N}}(%
\frac{t }{4N}x_{1},\ldots ,\frac{t}{4N}x_{n})\,d^{n}x  \notag \\
& =\frac{\Gamma (N_{\alpha })}{r^{N_{\alpha }-1}(4N)^{n}}\int_{\mathbb{R}
^{n}}f(x_{1},\ldots ,x_{n})\,d^{n}x\frac{1}{2\pi i}\int_{-i\infty
+0^{+}}^{-i\infty +0^{+}}e^{rt}t^{-(N_{\alpha }-n)}  \notag \\
& \times R_{n}^{LUE,\frac{1}{4N}}(\frac{t}{4N}x_{1},\ldots ,\frac{t}{4N}%
x_{n})dt.  \notag
\end{align}
Since $R_{n}^{LUE,\frac{1}{4N}}$ and $R_{n}^{\delta ,r}$ are continuous, we
get
\begin{align}
& R_{n}^{\delta ,r}(x_{1},\ldots ,x_{n})  \notag \\
& =\frac{\Gamma (N_{\alpha })}{r^{N_{\alpha }-1}(4N)^{n}}\frac{1}{2\pi i}%
\int_{-i\infty +0^{+}}^{-i\infty +0^{+}}e^{rt}t^{-(N_{\alpha }-n)}R_{n}^{LUE,%
\frac{1}{4N}}(\frac{t}{4N}x_{1},\ldots ,\frac{t}{4N}x_{n})dt  \notag \\
& =\frac{\Gamma (N_{\alpha })}{r^{N_{\alpha }-1}(4N)^{n}}\mathcal{L}%
^{-1}[t^{-(N_{\alpha }-n)}R_{n}^{LUE,\frac{1}{4N}}(\frac{t}{4N}x_{1},\ldots
, \frac{t}{4N}x_{n})](r).  \notag
\end{align}
This completes the proof.
\end{proof}

\begin{proof}(Theorem \ref{theorem for CLUE: bulk}: the hard edge)

Now we make use of the fact that $n$-point correlation function $R_{n}^{LUE,%
\frac{1}{4N}}(x_{1},\ldots ,x_{n})$ can be expanded as follows:
\begin{align}
& R_{n}^{LUE,\frac{1}{4N}}(x_{1},\ldots ,x_{n})=\prod_{i=1}^{n}x_{i}^{\alpha
}\exp \big(-4N\sum_{i=1}^{n}x_{i}\big)\sum_{l_{1},\ldots
,l_{n}=0}^{2N-2}c_{\{l_{1},\ldots ,l_{n}\}}^{(N)}x_{1}^{l_{1}}\cdots
x_{n}^{l_{n}}  \label{correlationfunctionexpansion} \\
& =\exp \big(-4N\sum_{i=1}^{n}x_{i}\big)\sum_{l_{1},\ldots
,l_{n}=0}^{2N-2}c_{\{l_{1},\ldots ,l_{n}\}}^{(N)}x_{1}^{l_{1}^{^{\prime
}}}\cdots x_{n}^{l_{n}^{^{\prime }}},  \notag
\end{align}
where $l_{i}^{^{\prime }}=l_{i}+\alpha ,i=1,\ldots ,n.$ It follows
from Proposition \ref{disintegrationforRTE} and
Eq.(\ref{correlationfunctionexpansion}) that
\begin{align}
&R_{n}^{\delta ,\frac{N+\alpha }{4}}(x_{1},\ldots ,x_{n})  \notag \\
& =\frac{4^{N_{\alpha }-1}\Gamma (N_{\alpha })}{(N+\alpha )^{N_{\alpha }-1}(4N)^{n}}%
\sum_{l_{1},\ldots ,l_{n}=0}^{2N-2}\frac{c_{\{l_{1},\ldots ,l_{n}\}}^{(N)}}{%
(4N)^{\sum l_{i}^{^{\prime }}}}x_{1}^{l_{1}^{^{\prime }}}\cdots
x_{n}^{l_{n}^{^{\prime }}}\mathcal{L}^{-1}[t^{-(N_{\alpha }-n-\sum
l_{i}^{^{\prime }})}](\frac{N+\alpha }{4}-\sum_{i=1}^{n}x_{i})  \notag \\
& =\theta (\frac{N+\alpha }{4}-\sum_{i=1}^{n}x_{i})\sum_{l_{1},\ldots
,l_{n}=0}^{2N-2}\frac{(\frac{N+\alpha }{4}-\sum_{i=1}^{n}x_{i})^{N_{\alpha
}-n-\sum l_{i}^{^{\prime }}-1}\frac{4^{N_{\alpha }-1}}{(N+\alpha
)^{N_{\alpha }-1}}}{(4N)^{\sum l_{i}^{^{\prime }}+n}}  \notag \\
& \frac{\Gamma (N_{\alpha })}{\Gamma (N_{\alpha }-n-\sum l_{i}^{^{\prime }})}%
c_{\{l_{1},\ldots ,l_{n}\}}^{(N)}x_{1}^{l_{1}^{^{\prime }}}\cdots
x_{n}^{l_{n}^{^{\prime }}}  \notag \\
& =\theta (\frac{N+\alpha }{4}-\sum_{i=1}^{n}x_{i})\sum_{l_{1},\ldots
,l_{n}=0}^{2N-2}(1-\frac{4}{N+\alpha }\displaystyle\sum_{i=1}^{n}x_{i})^{N_{%
\alpha }-n-\sum l_{i}^{^{\prime }}-1}  \label{expansionforFTE} \\
& \frac{\Gamma (N_{\alpha })}{(N_{\alpha })^{\sum l_{i}^{^{\prime
}}+n}\Gamma (N_{\alpha }-n-\sum l_{i}^{^{\prime }})}c_{\{l_{1},\ldots
,l_{n}\}}^{(N)}x_{1}^{l_{1}^{^{\prime }}}\cdots x_{n}^{l_{n}^{^{\prime }}}.
\notag
\end{align}%
With rescaling $x_{i}=\frac{t_{i}}{16N^{2}}$, we will deal with the
different factors in (\ref{expansionforFTE}). First, the $\theta $-function
term
\begin{equation}
\theta (\frac{N+\alpha }{4}-\frac{1}{N^{2}}\sum_{i=1}^{n}\frac{t_{i}}{16}%
)\longrightarrow 1  \label{thetaterm}
\end{equation}%
as $N\rightarrow \infty $. Since $\sum l_{i}^{^{\prime }}\leq
n(2N-2+\alpha ) $, the first  term in the sum yields
\begin{align}
& (1-\frac{4}{N^{2}(N+\alpha )}\displaystyle\sum_{i=1}^{n}\frac{t_{i}}{16}%
)^{N_{\alpha }-n-\sum l_{i}^{^{\prime }}-1}  \label{r-term} \\
& =1-\big(N_{\alpha }-n-\sum l_{i}^{^{\prime }}-1\big)\frac{4}{
N^{2}(N+\alpha )}\displaystyle\sum_{i=1}^{n}\frac{t_{i}}{16}+\frac{1}{N^{2}}%
C_{N}(t_{1},\cdots,t_{n})
\end{align}
Here $C_{N}(t_{1},\ldots ,t_{n})$ is uniformly bounded in all $t_{1}\in
(0,A],\cdots ,t_{n}\in (0,A]$ for given $A>0$.

Before evaluating the factor containing $\Gamma$-functions, we
introduce a lemma about ratio of two Gamma functions, due to Tricomi
and Erd\'{e}lyi \cite{tricomi}, see also Copson's book
\cite{copson}.
\begin{lemma}
\label{ratiooftwogamma} Let $a>0$, for sufficiently large $x$ the following
expansion holds:
\begin{align}
\label{ratiooftwoG}
\frac{\Gamma(x)}{x^{a}\Gamma(x-a)}=\sum_{s=0}^{\infty}\frac{a(a-1)\cdots
(a-s+1)}{s! x^{s}} B^{(a+1)}_{s}(0).
\end{align}
Here $B^{(a)}_{s}(x)$ is N\"{o}rlund's generalized Bernoulli polynomial in $%
x $ of degree $s$, and $B^{(a)}_{s}(0)$ is also a polynomial in $a$
of degree $s$.
\end{lemma}\noindent Note that there are only finite terms in the sum of the right-hand side of Eq.(\ref{ratiooftwoG}) if $\alpha$ is an integer.

For the convenience, we write
\begin{equation*}
L_{s}(a)=a(a-1)\cdots (a-s+1)B_{s}^{(a+1)}(0).
\end{equation*}%
It is a polynomial in $a$ of degree $2s$. By Lemma \ref{ratiooftwogamma},
the factor containing $\Gamma $-functions reads:
\begin{equation}
\frac{\Gamma (N_{\alpha })}{(N_{\alpha })^{\sum l_{i}^{^{\prime }}+n}\Gamma
(N_{\alpha }-n-\sum l_{i}^{^{\prime }})}=\sum_{s=0}^{\infty }\frac{1}{%
s!(N_{\alpha })^{s}}L_{s}(n+\sum l_{i}^{^{\prime }}).  \label{gammaterm}
\end{equation}%
Since $B_{0}^{(a)}(x)=1$, we have $L_{0}(n+\sum l_{i}^{^{\prime }})=1$. However, if $%
\sum l_{i}^{^{\prime }}\sim N$, the terms on the right-hand side of
Eq.(\ref{gammaterm})
\begin{equation}
\frac{1}{(N_{\alpha })^{s}}L_{s}(n+\sum l_{i}^{^{\prime }})
\end{equation}%
approach to a finite non-zero number as $N\rightarrow \infty $, not
sub-leading. With the aid of multiplication and partial differential
operators, we can rewrite
\begin{equation}
L_{s}(n+\sum l_{i}^{^{\prime }})\,t_{1}^{l_{1}^{^{\prime }}}\cdots
t_{n}^{l_{n}^{^{\prime }}}=L_{s}(\sum_{i=1}^{n}\partial
_{t_{i}}t_{i})\,t_{1}^{l_{1}^{^{\prime }}}\cdots t_{n}^{l_{n}^{^{\prime }}}.
\label{operatorsum}
\end{equation}%
Combing (\ref{thetaterm}), (\ref{r-term}), (\ref{gammaterm}) and (\ref%
{operatorsum}), it follows from (\ref{expansionforFTE}) that
\begin{align}
& \frac{1}{(16N^{2})^{n}}R_{n}^{\delta ,\frac{N+\alpha }{4}}(\frac{t_{1}}{%
16N^{2}},\cdots ,\frac{t_{n}}{16N^{2}})\notag \\
& =\frac{1}{(16N^{2})^{n}}\sum_{l_{1},\ldots ,l_{n}=0}^{2N-2}\Big(1-(N_{\alpha
}-n-\sum l_{i}^{^{\prime }}-1)\frac{4}{N^{2}(N+\alpha )}\sum_{i=1}^{n}\frac{%
t_{i}}{16} \notag\\
& +\frac{1}{N^{2}}C_{N}(t_{1},\cdots ,t_{N})\Big)\sum_{s=0}^{\infty }\frac{L_{s}(n+\sum l_{i}^{^{\prime }})}{%
s!(N_{\alpha })^{s}}c_{\{l_{1},\ldots ,l_{n}\}}^{(N)}(\frac{t_{1}}{16N^{2}}%
)^{l_{1}^{^{\prime }}}\cdots (\frac{t_{n}}{16N^{2}})^{l_{n}^{^{\prime }}} \notag\\
& =I-\frac{4(N_{\alpha }-1)}{N^{2}(N+\alpha )}(\sum_{i=1}^{n}\frac{t_{i}}{16}%
)I+\frac{4}{N^{2}(N+\alpha )}(\sum_{i=1}^{n}\frac{t_{i}}{16}) \\
& \times \frac{1}{(16N^{2})^{n}}\sum_{l_{1},\ldots
,l_{n}=0}^{2N-2}(\sum_{i=1}^{n}\partial _{t_{i}}t_{i})\sum_{s=0}^{\infty }%
\frac{L_{s}(\sum_{i=1}^{n}\partial _{t_{i}}t_{i})}{s!(N_{\alpha })^{s}}%
c_{\{l_{1},\ldots ,l_{n}\}}^{(N)}(\frac{t_{1}}{16N^{2}})^{l_{1}^{^{\prime
}}}\cdots (\frac{t_{n}}{16N^{2}})^{l_{n}^{^{\prime }}}\notag \\
& +\frac{1}{(16N^{2})^{n}}\sum_{l_{1},\ldots ,l_{n}=0}^{2N-2}\frac{1}{N^{2}}%
C_{N}(t_{1},\cdots ,t_{n})\sum_{s=0}^{\infty }\frac{L_{s}(\sum_{i=1}^{n}\partial _{t_{i}}t_{i})}{s!(N_{\alpha })^{s}} \notag\\
& \times c_{\{l_{1},\ldots ,l_{n}\}}^{(N)}(\frac{t_{1}}{16N^{2}}%
)^{l_{1}^{^{\prime }}}\cdots
(\frac{t_{n}}{16N^{2}})^{l_{n}^{^{\prime }}}\doteq
I+I_{2}+I_{3}+I_{4}\notag
\end{align}
where
\begin{equation}\label{I limit}
I=\sum_{s=0}^{\infty }\frac{L_{s}(\sum_{i=1}^{n}\partial _{t_{i}}t_{i})}{%
s!(N_{\alpha })^{s}}\frac{1}{(16N^{2})^{n}}\sum_{l_{1},\ldots
,l_{n}=0}^{2N-2}c_{\{l_{1},\ldots ,l_{n}\}}^{(N)}(\frac{t_{1}}{16N^{2}}%
)^{l_{1}^{^{\prime }}}\cdots (\frac{t_{n}}{16N^{2}})^{l_{n}^{^{\prime }}}.
\end{equation}%
On the other hand, we know from the expansion (\ref%
{correlationfunctionexpansion}) for the correlation function of the
LUE that
\begin{align}
& \lim_{N\rightarrow \infty }\frac{1}{(16N^{2})^{n}}R_{n}^{LUE,\frac{1}{4N}}(\frac{t_{1}}{%
16N^{2}},\cdots ,\frac{t_{n}}{16N^{2}})  \label{finalexpansion} \\
& =\lim_{N\rightarrow \infty }\frac{1}{(16N^{2})^{n}}\exp \big(%
-\sum_{i=1}^{n}\frac{t_{i}}{4N}\big)\sum_{l_{1},\ldots
,l_{n}=0}^{2N-2}c_{\{l_{1},\ldots ,l_{n}\}}^{(N)}(\frac{t_{1}}{16N^{2}}%
)^{l_{1}^{^{\prime }}}\cdots (\frac{t_{n}}{16N^{2}})^{l_{n}^{^{\prime }}}
\notag \\
& =\lim_{N\rightarrow \infty }\frac{1}{(16N^{2})^{n}}\sum_{l_{1},\ldots
,l_{n}=0}^{2N-2}c_{\{l_{1},\ldots ,l_{n}\}}^{(N)}(\frac{t_{1}}{16N^{2}}%
)^{l_{1}^{^{\prime }}}\cdots (\frac{t_{n}}{16N^{2}})^{l_{n}^{^{\prime }}}
\notag \\
& =\det [K_{\textrm{J}_{\alpha }}(t_{i},t_{j})]_{i,j=1}^{n},  \notag
\end{align}%
uniformly for $t_{1},\ldots, t_{n}$ in bounded subsets of
$(0,\infty)$. Comparing (\ref{I limit}) and (\ref{finalexpansion}),
the terms where $s>0$ in the sum of (\ref{I limit}) vanish. Hence we
get
\begin{equation}\label{limit}
\lim_{N\rightarrow \infty }I=\det [K_{\textrm{J}_{\alpha
}}(t_{i},t_{j})]_{i,j=1}^{n}.
\end{equation}%
Thus we have that $I_{2}\rightarrow 0$ as $N\rightarrow \infty .$ Note that
\begin{align*}
& \lim_{N\rightarrow \infty }I_{3}=\lim_{N\rightarrow \infty }\frac{4}{%
N^{2}(N+\alpha )}(\sum_{i=1}^{n}\frac{t_{i}}{16})\sum_{s=0}^{\infty }\frac{%
(\sum_{i=1}^{n}\partial _{t_{i}}t_{i})L_{s}(\sum_{i=1}^{n}\partial
_{t_{i}}t_{i})}{s!(N_{\alpha })^{s}}  \\
& \times \frac{1}{(16N^{2})^{n}}\sum_{l_{1},\ldots
,l_{n}=0}^{2N-2}c_{\{l_{1},\ldots ,l_{n}\}}^{(N)}(\frac{t_{1}}{16N^{2}}%
)^{l_{1}^{^{\prime }}}\cdots (\frac{t_{n}}{16N^{2}})^{l_{n}^{^{\prime }}}=0
\end{align*}%
We also notice the following fact: for any sequences $\{a_{i}\}_{i=1}^{N}$ and $%
\{b_{i}^{(N)}\}_{i=1}^{N}$, if $a_{1}+\cdots +a_{N}\rightarrow C $
as $N\rightarrow \infty $, and $b_{i}^{(N)}=O(\frac{1}{N^{2}}),$
then
\begin{equation}
|a_{1}b_{1}^{(N)}+a_{2}b_{2}^{(N)}+\cdots +b_{N}^{(N)}a_{N}|\leq
\sum_{i=1}^{N}|S_{i}-S_{i-1}|\,|b_{i}^{(N)}|=O\left(\frac{1}{N}\right)
\end{equation}
where $S_{i}=a_{1}+a_{2}+\cdots +a_{i}.$ Thus we have
\begin{equation*}
|I_{4}|\leq C\,n(2N-2+\alpha )O(\frac{1}{N^{2}})=O(\frac{1}{N}),
\end{equation*}%
uniformly for $t_{1},\ldots, t_{n}$ in bounded subsets of
$(0,\infty)$. Here we have used the fact of   (\ref{limit}).

This completes the proof.
\end{proof}

\section{proof of theorem \protect\ref{theorem for BTLUE: bulk}}

 First we give a representation of correlation functions for the BTLUE in terms of these for the FTLUE.
\begin{proposition}
\label{relationsforBTLUE} Let $R_{n}^{\theta,r}$ and $R_{n}^{\delta
,r}$ be the $n$-point correlation functions for the BTLUE and FTLUE
respectively, then we have the following relation
\begin{equation}
R_{n}^{\theta,r}(x_{1},\ldots
,x_{n})=\int_{0}^{1}N_{\alpha}\,u^{N_{\alpha}-1}\,
\frac{1}{u^{n}}R_{n}^{\delta ,r}(\frac{x_{1}}{u},\ldots
,\frac{x_{n}}{u})\,du , \label{BTLUE and FTLUE}
\end{equation}
where $N_{\alpha}=N(N+\alpha)$.
 \end{proposition}

\begin{proof} It suffices to prove
\begin{equation}
R_{n}^{\theta,r}(x_{1},\ldots
,x_{n})=\int_{0}^{r}\frac{N_{\alpha}}{r^{N_{\alpha}}}u^{N_{\alpha}-1}\,
(\frac{r}{u})^{n}R_{n}^{\delta ,r}(\frac{r}{u}x_{1},\ldots
,\frac{r}{u}x_{n})\,du.
\end{equation}
For every $u>0$, let
 \begin{equation}
\Delta_{N}(u)=\{(x_{1},\ldots,x_{N})| \sum_{j=1}^{N}x_{j}=u,
x_{j}\geq 0, j=1,\ldots, N\}
\end{equation}
be  a simplex in $\mathbb{R}^{N}$, which carries the volume element
induced by the standard Euclidean metric on $\mathbb{R}^{N}$,
denoted by $u^{N-1}d \,\sigma_{N}$. For $h\in
L^{\infty}(\mathbb{R}^{N})$, let $<h(\cdot)>_{\theta}$ and
$<h(\cdot)>_{\delta}$ denote  the ensemble average taken in the
BTLUE and the FTLUE, respectively. From (\ref{probability of fixed
trace lauguerre}) and (\ref{probability of bounded trace
lauguerre}), we have
\begin{flalign*}
&<h(\cdot)>_{\theta}\\
&= \frac{1}{Z_{\theta }^{r}}\int_{0}^{r}u^{N-1}d\,u
\int_{\Delta_{N}(u)}h(x_{1},\ldots,x_{N})\prod_{1\leq i<j\leq
N}|x_{i}-x_{j}|^{2}\prod_{i=1}^{N}x_{i}^{\alpha }d \,\sigma_{N}\\
&= \frac{1}{Z_{\theta }^{r}}\int_{0}^{r}u^{N_{\alpha}-1}d\,u
\int_{\Delta_{N}(1)}h(u x_{1},\ldots,u x_{N})\prod_{1\leq i<j\leq
N}|x_{i}-x_{j}|^{2}\prod_{i=1}^{N}x_{i}^{\alpha }d \,\sigma_{N}
\end{flalign*}
and
\begin{flalign*}
&<h(a\,\cdot)>_{\delta}\\
&=\frac{1}{Z_{\delta }^{r}} \int_{\Delta_{N}(r)}r^{N-1}\,h(a
x_{1},\ldots,a x_{N})\prod_{1\leq i<j\leq
N}|x_{i}-x_{j}|^{2}\prod_{i=1}^{N}x_{i}^{\alpha } d \,\sigma_{N}\\
&=\frac{1}{Z_{\delta }^{r}}
\int_{\Delta_{N}(1)}r^{N_{\alpha}-1}\,h(a r x_{1},\ldots,a r
x_{N})\prod_{1\leq i<j\leq
N}|x_{i}-x_{j}|^{2}\prod_{i=1}^{N}x_{i}^{\alpha } d \,\sigma_{N}.
\end{flalign*}
Choose $a=\frac{u}{r}$, we get
\begin{equation}
<h(\cdot)>_{\theta} =\frac{Z_{\delta}^{r}} {Z_{\theta }^{r}}
\int_{0}^{r}(\frac{u}{r})^{N_{\alpha}-1}<h(\frac{u}{r}\,\cdot)>_{\delta}
d\,u.
\end{equation}
Setting $h\equiv 1$, we get the ratio of the partition functions
$Z_{\delta}^{r}$ and $Z_{\theta }^{r}$. Substituting this ratio, we
then obtain
\begin{equation}
<h(\cdot)>_{\theta}= \int_{0}^{r}\frac{N_{\alpha}}{r^{N_{\alpha}}}
u^{ N_{\alpha}-1 } <h(\frac{u}{r}\,\cdot)>_{\delta} d\,u.
\end{equation}
In particular, taking
\begin{equation}
h(x_{1},\ldots,x_{N})=\sum_{1\leq i_{1}<\cdots <i_{n}\leq N
}f(x_{i_{1}},\ldots,x_{i_{n}}),
\end{equation}
we have
\begin{flalign*}
&\int_{\mathbb{R}^{n}}f(x_{1},\ldots,x_{n})R_{n}^{\theta,r}(x_{1},\ldots,x_{n})\,d^{n}x\nonumber \\
&=\int_{0}^{r}\frac{N_{\alpha}}{r^{N_{\alpha}}} u^{ N_{\alpha}-1
}d\,u
\int_{\mathbb{R}^{n}}f(\frac{u}{r}x_{1},\ldots,\frac{u}{r}x_{n})\,
R_{n}^{\delta,r}(x_{1},\ldots,x_{n})\,d^{n}x\nonumber \\
&=\int_{\mathbb{R}^{n}}f(x_{1},\ldots,x_{n})\,d^{n}x
\int_{0}^{r}\frac{N_{\alpha}}{r^{N_{\alpha}}}u^{N_{\alpha}-1}\,
(\frac{r}{u})^{n}R_{n}^{\delta ,r}(\frac{r}{u}x_{1},\ldots
,\frac{r}{u}x_{n})\,du.
\end{flalign*}
Since $R_{n}^{\theta,r}$ and $R_{n}^{\delta ,r}$ are both
continuous, we complete the proof.
\end{proof}

Next,   we notice a ``sharp" concentration phenomenon along the
radial coordinate between correlation functions of the BTLUE and
FTLUE. Although its proof  is simple, the following lemma plays a
crucial role in dealing with local statistical properties of the
eigenvalues between the fixed and bounded ensembles.

\begin{lemma}
\label{concentrationlemma2}Let $\{b_{N}\}$ be a sequence such that
$b_{N}\rightarrow 0$ but $N^{2} b_{N}\rightarrow\infty$ as
$N\rightarrow\infty$, then we have
\begin{equation}\int_{0}^{1}N_{\alpha}\,u^{N_{\alpha}-1}\,
\,du=\int_{u_{-}}^{1} N_{\alpha}\,u^{N_{\alpha}-1}\,du+ e^{-N^{2}
b_{N}(1+o(1))},
\end{equation}
 where
$u_{-}=1-b_{N}$.
 \end{lemma}
\begin{proof}
\begin{flalign*}
&\int_{0}^{u_{-}} N_{\alpha}\,u^{N_{\alpha}-1}
\,du=(1-b_{N})^{N_{\alpha}}=e^{N_{\alpha}\ln
(1-b_{N})}=e^{N_{\alpha}\big(-b_{N}+O(b^{2}_{N})\big)}=e^{-N^{2}
b_{N}(1+o(1))}.
\end{flalign*}
This completes the proof.
\end{proof}

\begin{remark}
In Lemma \ref{concentrationlemma2}, let us take $b_{N}=N^{-\kappa},
\kappa\in (0,2)$. Since the ``rate" index $\kappa$ can be chosen
larger than 1 while the scaling in the bulk is proportional to
$N^{-1}$ and at the soft edge of the spectrum is proportional to
$N^{-2/3}$, in principle  we can prove all local statistical
properties of the eigenvalues between  the fixed and bounded trace
ensembles are identical in the limit. Such arguments apply to  the
equivalence of ensembles between the fixed trace and bounded trace
ensembles with monomial potentials, where we exploit some
homogeneity of the monomial potentials.
\end{remark}

Before we prove Theorem \protect\ref{theorem for BTLUE: bulk}, let
us prove the claimed result in Sect. 1: limit global density for the
BTLUE is also  Marchenko-Pastur law.

\begin{theorem}
\label{MPlawforBTLUE} Let $R_{1}^{\theta,r}$ be the 1-point
correlation function of the BTLUE, for any $f\in
L^{\infty}(\mathbb{R})\bigcap C_{lip}(\mathbb{R})$ where the set
$C_{lip}(\mathbb{R})$ denotes  all Lipschitz continuous functions on
$\mathbb{R}$, we have
\begin{align}
\lim_{N\rightarrow \infty }\int_{\mathbb{R}}f(x) \frac{1}{N}
R_{1}^{\theta,\frac{N+\alpha }{4}}(x) \,d x  =\int_{\mathbb{R}}f(x)
\psi (x) \,d x
\end{align}
where $\psi (x)$ is the Marchenko-Pastur law.
\end{theorem}

\begin{proof}
By Proposition \ref{relationsforBTLUE},
\begin{align*}
\int_{\mathbb{R}}f(x)& \frac{1}{N} R_{1}^{\theta,\frac{N+\alpha
}{4}}(x) \,d x  =\int_{\mathbb{R}}
\int_{0}^{1}N_{\alpha}\,u^{N_{\alpha}-1}\,f(x)
\frac{1}{N}\frac{1}{u}R_{1}^{\delta ,\frac{N+\alpha
}{4}}(\frac{x}{u})\,du  \,d x\\
&=\int_{\mathbb{R}} \int_{0}^{1}N_{\alpha}\,u^{N_{\alpha}-1}\,f(u x)
\frac{1}{N}R_{1}^{\delta ,\frac{N+\alpha
}{4}}(x)\,du  \,d x\\
&=\int_{\mathbb{R}} \left(\int_{0}^{u_{-}}+\int_{u_{-}}^{1}\right)
N_{\alpha}\,u^{N_{\alpha}-1}\,f(u x) \frac{1}{N}R_{1}^{\delta
,\frac{N+\alpha }{4}}(x)\,du  \,d x\doteq I_{1}+I_{2}.
\end{align*}
By Lemma \ref{concentrationlemma2},
\begin{align*}
I_{1}&\leq \|f\|_{\infty}\int_{\mathbb{R}} \int_{0}^{u_{-}}
N_{\alpha}\,u^{N_{\alpha}-1} \frac{1}{N}R_{1}^{\delta
,\frac{N+\alpha }{4}}(x)\,du  \,d x\\
&=\|f\|_{\infty} \int_{0}^{u_{-}} N_{\alpha}\,u^{N_{\alpha}-1} \,du
\int_{\mathbb{R}} \frac{1}{N}R_{1}^{\delta ,\frac{N+\alpha }{4}}(x)
\,d x\\
&=\|f\|_{\infty} \,e^{-N^{2} b_{N}(1+o(1))}\rightarrow 0
\end{align*}
as $N\longrightarrow \infty$. On the other hand,
\begin{align*}
I_{2}&=\int_{\mathbb{R}} \int_{u_{-}}^{1}
N_{\alpha}\,u^{N_{\alpha}-1}\left(\big(f(u x)-f(x)\big)+f(x)\right)
\frac{1}{N}R_{1}^{\delta ,\frac{N+\alpha }{4}}(x)\,du  \,d x\doteq I_{21}+I_{22}.\\
\end{align*}
Since $f\in C_{lip}(\mathbb{R})$, we have
\begin{align*}
I_{21}\leq  (1-u_{-}) L\, \int_{u_{-}}^{1}
N_{\alpha}\,u^{N_{\alpha}-1} du \int_{\mathbb{R}}|x|
\frac{1}{N}R_{1}^{\delta ,\frac{N+\alpha }{4}}(x) \,d x\rightarrow 0\\
\end{align*}
for some constant $L$. Here we have used the fact \begin{align*}
\lim_{N\rightarrow \infty }\int_{\mathbb{R}}|x|
\frac{1}{N}R_{1}^{\delta ,\frac{N+\alpha }{4}}(x) \,d
x=\int_{\mathbb{R}}|x|
\psi (x) \,d x.\\
\end{align*}
Again by Lemma \ref{concentrationlemma2},  \begin{align*}
\lim_{N\rightarrow \infty }I_{22}=\int_{\mathbb{R}}f(x)
\psi (x) \,d x.\\
\end{align*}
This completes the proof.
\end{proof}

Now we turn to the proof of Theorem \protect\ref{theorem for BTLUE:
bulk}. \begin{proof} The proof is very similar to that of the soft
edge of the spectrum in Theorem \ref{theorem for CLUE: bulk}, we
only point out some different places in the bulk case.

 In Lemma
\ref{concentrationlemma2}, choose $b_{N}=N^{-\kappa}, \,\kappa \in
(1,2)$. The change of variables corresponding to
(\ref{changeofvariables}) reads:
\begin{equation}
t_{i}=(u-1)N\, x \psi(x)+u y_{i},i=1,\ldots ,n
\end{equation}
where fixed $x\in (0,1)$. The condition that $b_{N}=N^{-\kappa},
\,\kappa
 \in (1,2)$ ensures $(1-u)N\leq  N^{-\kappa+1}\rightarrow 0$ as
 $N\longrightarrow \infty$ for $u\in [u_{-},1]$. On the other hand, by Theorem \ref{theorem for CLUE: bulk}, the following fact similar to
 Lemma \ref{boundlemma} is obvious:
 for any fixed $R>0$,
\begin{equation}
\frac{1}{(N\psi(x))^{n}}\int_{B_{R}}R_{n}^{\delta ,\frac{N+\alpha
}{4}}(x+ \frac{y_{1}}{N\psi(x)},\cdots
,x+\frac{y_{n}}{N\psi(x)})d^{n}y\leq C_{R}.
\end{equation}
Here  $B_{R}$ is the ball of the
radius $R$ in $\mathbb{R}^{n}$ centered at zero, and  $%
C_{R}$ is a constant.

Using Proposition \ref{relationsforBTLUE} and Theorem \ref{theorem
for CLUE: bulk}, we complete the proof after a similar procedure.
\end{proof}


\begin{thebibliography}{99}
\bibitem{acmv} Akemann, G., Cicuta, G.M., Molinari, L., Vernizzi, G.:
Compact support probability distributions in random matrix theory. Phys.
Rev. E \textbf{59}(2), 1489--1497 (1999)

\bibitem{av} Akemann, G., Vernizzi, G.: Macroscopic and microscopic
(non-)universality of compact support random matrix theory. Nucl. Phys. B
\textbf{583} (3), 739--757 (2000)

\bibitem{toda2} Adachia, S.,   Toda, M.,  Kubotanic, H.: Random matrix theory of
singular values of rectangular complex matrices I: Exact formula of
one-body distribution function in fixed-trace ensemble. Annals of
Physics, Vol. 324, Issue 11,  2278-2358 (2009)
\bibitem{bengtsson} Bengtsson, I., \.{Z}yczkowski,   K.: {\it Geometry of
Quantum States}. Cambridge Univ. Press, New York, 2006

\bibitem{copson}
Copson, E.T.: {\it Asymptotic Expansions}. Cambridge University
Press, Cambridge, 1965

\bibitem{deift} Deift. P, Orthogonal polynomials and random matrices: a
Riemann-Hilbert approach. Courant Lecture Notes in Mathematics, Vol. 3. N.
Y.: New York University Courant Institute of Mathematical Sciences, 1999


\bibitem{DG}
Deift, P., Gioev, D.: Universality at the edge of the spectrum for
unitary, orthogonal, and symplectic ensembles of random matrices.
Comm. Pure Appl. Math. \textbf{60}, 867--910 (2007)

\bibitem{dkmvz} Deift, P., Kriecherbauer, T., McLaughlin, K. T.-R.,
Venakides, S., Zhou, X.: Uniform asymptotics for polynomials
orthogonal with respect to varying
   exponential weights and applications to universality questions in random
   matrix theory. Comm. Pure Appl. Math. \textbf {52}(11), 1335--1425 (1999)







\bibitem{forrester1} Forrester, P. J.: The spectrum edge of random matrix
ensembles. Nucl. Phys. B \textbf{402},709--728 (1993)

\bibitem{feller}Feller, W.: An introduction to probability theory and its application. Vol II, Second edition, New York; John Wiley and Sons Inc.,1971.

\bibitem{gg} G$\ddot{o}$tze, F., Gordin, M. : Limit correlation functions for
fixed trace random matrix ensembles. Commun. Math. Phys.
\textbf{281}, 203-229 (2008)

\bibitem{ggl} G\"otze, F., Gordin, M., Levina, A.: Limit correlation
function at zero for fixed trace random matrix ensembles.(Russian) Zap.
Nauchn. Sem. S.-Peterburg. Otdel. Mat. Inst. Steklov. (POMI) \textbf{341},
68--80 (2007); translation to appear in J. Math. Sci. (N. Y.) \textbf{145}%
(3) (2007)

\bibitem{HLW}  Hayden, P.,  Leung, D. W.,  Winter, A.: Aspects of generic
entanglement. Comm. Math. Phys. 265: 95-117 (2006)

\bibitem{Johnstone} Johnstone, I. M.: On the distribution of the largest
eigenvalue in principal components analysis. Ann. Stat. \textbf{29},
295--327 (2001)


\bibitem{LZ} Liu, D.-Z., Zhou, D.-S.: Universality at Zero and the Spectrum
Edge for Fixed Trace Gaussian $\beta$-Ensembles of Random Matrices.
arXiv: 0905.4932v1
\bibitem{Majumdar} Majumdar, S.N, Bohigas, O., Lakshminarayan, A.: Exact Minimum
Eigenvalue Distribution of an Entangled Random Pure State. J. Stat.
Phys. \text{131}, 33--49(2008)
\bibitem{mehta} Mehta, M. L.: \textit{Random matrices}. 3rd edn. Pure
and Applied Mathematics (Amsterdam), 142. Amsterdam: Elsevier/Academic
Press, 2004



\bibitem{Michael} Hall, Michael J. W.: Random quantum correlations and
density operator distributions. Phy. Let.A
\textbf{242},123--129(1998)

\bibitem{MP}  Mar\v{c}enko, V. A., Pastur, and L. A.: Distributions of
eigenvalues of some sets of random matrices. Math. USSR-Sb., \textbf{1}%
:507--536(1967)

\bibitem{Neilsen} Neilsen, M.A., Chuang, I.L.: Quantum Computation and
Quantum Information. Cambridge University Press, Cambridge (2000)

\bibitem{Nechita} Nechita, I.: Asymptotics of Random Density Matrices. Ann.
Henri. Poincar\'{e} 8, 1521--1538(2007)

\bibitem{Nagao} Nagao, T.,   Wadati, M.: Correlation functions of random matrix
ensembles related to classical orthogonal polynomials. J. Phys. Soc.
Japan, \textbf{60}(10),3298--3322(1997)


\bibitem{Page} Page,D. N., Average entropy of a subsystem, Phys. Rev. Lett.
\textbf{71}, 1291 - 1294 (1993)


\bibitem{szego} Szeg\"{o}, G.: \textit{Orthogonal polynomials}.
American Mathematical Society, New York, 1939
\bibitem{So} Sommers, H.-J., Zyczkowski, K.: Statistical properties of
random density matrices. J. Phys. A: Math. Gen. \textbf{37},
8457--8466(2004)

\bibitem{tw1} Tracy, C. A., Widom, H.: Level-spacing distributions and the
Airy kernel. Commun. Math. Phys. \textbf{159}, 151--174 (1994)

\bibitem{tw2} Tracy, C. A., Widom, H.: Level spacing distributions and the
Bessel kernel. Comm. Math. Phys., \textbf{161}(2):289--309(1994)

\bibitem{Vanlessen}  Vanlessen, M.: Strong Asymptotics of Laguerre-Type
Orthogonal Polynomials and Application in Random Matrix Theory.
Constr. Approx \textbf{25}, 125--175(2007)

\bibitem{widim} Widom, H.: On the relation between orthogonal, symplectic and
unitary matrix ensembles. J. Statist. Phys., \textbf{94}%
(3--4):347--363(1999)

\bibitem{Werner} Werner, R. F., Quantum states with Einstein-Podolsky-Rosen
correlations admitting a hidden-variable model. Phys. Rev. A
\textbf{40}, 4277--4281(1989)
\bibitem{tricomi}
Tricomi, F.G.,  Erd\.{e}lyi, A: The asymptotic expansion of a ratio
of gamma functions. Pacific J. Math. 1, No. 1, 133--142 (1951)

\bibitem{Zyczkowski2001} Zyczkowski, K., Sommers, H.-J.: Induced measures in
the space of mixed quantum states. J. Phys. A: Math. Gen. \text{34},
7111--7125 (2001)
\end{thebibliography}
\end{document}